\documentclass{article}
\pdfoutput=1
\usepackage[round]{natbib}
\usepackage{amsmath,amsfonts}
\usepackage{amsthm}
\usepackage[usenames]{color}
\usepackage{graphicx}
\usepackage[margin=2.5cm]{geometry}
\usepackage{algorithm}
\usepackage[noend]{algpseudocode}
\usepackage{array}
\newcolumntype{P}[1]{>{\centering\arraybackslash}p{#1}}
\usepackage{textgreek}
\newtheorem{proposition}{Proposition}

\newcommand{\bx}{\mathbf{x}}
\newcommand{\mub}{\bar{\mu}}
\newcommand{\nub}{\bar{\nu}}
\newcommand{\bmu}{\boldsymbol{\mu}}
\newcommand{\ab}{\bar{a}}
\newcommand{\bb}{\bar{b}}
\newcommand{\mb}{\bar{m}}
\newcommand{\cb}{\bar{c}}
\newcommand{\M}{\mathbf{m}}

\newcommand{\etamax}{\eta_{\max}}
\newcommand{\lambdamax}{\lambda_{\max}}
\newcommand{\SStat}{\mathcal{A}}
\newcommand{\SSbar}{\bar{\mathcal{A}}}
\newcommand{\zero}{\mathbf{0}}
\newcommand{\unit}{\mathbf{1}}
\newcommand{\mumin}{\mu_{\min}}
\newcommand{\Umax}{u_{\max}}
\newcommand{\micro}{\textmu}

\begin{document}
\title{Motor Unit Number Estimation via Sequential Monte Carlo}
\author{
  Simon Taylor, Chris Sherlock, Gareth Ridall and Paul Fearnhead
}
\date{11th April 2018}

\maketitle

\section*{Abstract}
A change in the number of motor units that operate a particular
muscle is an important indicator for the progress of a neuromuscular disease
and the efficacy of
a therapy. Inference for realistic statistical models of
the typical data produced
when testing muscle function is difficult, and estimating the number
of motor units from these data 
 is an ongoing statistical challenge.
We consider a set of models for the data, each with a different
number of working motor units, and present a novel method for Bayesian inference, based on sequential
Monte Carlo, which provides estimates of the marginal
likelihood and, hence, a posterior probability for each model.
To implement this approach in practice we require sequential Monte
Carlo methods that have excellent computational and Monte Carlo
properties. We achieve this by leveraging the conditional independence
structure in the model, where given knowledge of which motor units fired as a result of a particular stimulus,
parameters that specify the size of each unit's response are independent of the 
parameters defining the probability that a unit will respond at
all. The scalability of our methodology relies on the natural
conjugacy structure that we create for the former and an enforced, approximate conjugate structure
for the latter. A simulation study demonstrates the
accuracy of our method, and inferences are consistent across two
different datasets arising from the same rat tibial muscle.

\section*{Keywords}
Motor Unit Number Estimation; Sequential Monte Carlo; Model Selection

\newpage

\section{Introduction} \label{sec:Intro}

Motor unit number estimation (MUNE) is a continuing challenge for
clinical neurologists. An ability to determine the number of motor units (MUs)
that operate a particular muscle provides important insights into the
progression of various neuromuscular ailments such as amyotrophic lateral
sclerosis \citep{She06, Bro07}, and aids the assessment of the
efficacy of potential therapy treatments \citep{Cas10}.

A MU is the fundamental component of the neuromuscular system and 
 consists of a single motor neuron and
the muscle fibres whos contraction it governs. Restriction to a MU's
operation may be a result of impaired communication between the motor
neuron and muscle fibres, abnormaility in their function, or atrophy of
either cell type. A direct investigation into the number of MUs via a biopsy, for example, is not helpful since this only determines the presence of each MU, not its functionality.

Electromyography (EMG) provides a set of electrical stimulii of
varying intensity to a group of motor neurons; each stimulus
artificially induces a twitch in the targeted muscle, providing an
\textit{in situ} measurement of the functioning of the MUs. The effect on the
muscle may be measured by recording either the minute variation in muscle
membrane potential or the physical force the muscle exerts
\citep{Maj05}. The generic methods developed in this article are
applicable to either type of measurement. Since our data consist of whole muscle
twitch force (WMTF) measurements we henceforth describe the response
in these terms. 
 In a healthy subject, the
stimulus-response curve is typically sigmoidal \citep{Hen06},
illustrating the smooth recruitment of additional MUs as the stimulus
increases; however, the relatively low number of MUs in a patient with impaired muscle function may manifest within the stimulus-response
relationship as large jumps in WMTF measurements.

Figure \ref{fig:RatData} shows the two data sets that will be
described and analysed in detail 
in Section \ref{sec:CaseStudy}, with the large jumps clearly
visible. The histograms of absolute differences in response for
adjacent stimuli show two main modes, one, near 0\,mN, corresponding to
noise and the other, around 40\,mN indicating that different MUs fired. The noise arises primarily because of small 
variations in the
contribution to the WMTF provided by any particular MU, whenever it
fires. The second general source of noise, visible in isolation at very low
stimuli when no MUs are firing, is called the baseline noise. This
arises from 
respiration movements and pulse pressure waves,
and particular care is taken to minimise such influences, for example by  earthing the subject and equipment, restraining the limb, digitally resetting the force signal prior to each stimulus, synchronising stimuli with the pulse cycle and using highly sensitive measurement devices.

MUNE uses the
observed stimulus-response pattern to estimate the number of
functioning MUs. 
Techniques for MUNE generally form two classes: the average and
comprehensive approaches. The most common averaging approach is the
incremental technique of \cite{McC71}, which 
assumes that the MUs can be characterised by an `average' MU
with a particular single motor unit twitch force (MUTF), estimated as
the average of the magnitudes of the observed stepped increases in twitch
force. A large
stimulus, known as the supramaximal stimulus, is applied in order to
cause all MUs to react. The quotient of the WMTF arising from the supramaximal
stimulus and the average MUTF provides a count estimate. However, there is no guarantee that a particular single-stepped increase in
response corresponds to a new, previously latent, MU, since it may
instead be due to a phenomenon called alternation \citep{Bro76}. This
occurs when two or more MUs have similar activation thresholds such
that different combinations of MUs may fire in reaction to two
identical stimuli. Consequently, the incremental technique tends to
underestimate the average MUTF and hence overestimate the  number of
MUs. A number of improvements both experimentally
\citep[e.g.]{Kad76,Sta94} and empirically \citep[e.g.]{Dau95,Maj07}
have been proposed to try to deal with the alternation problem but,
despite these improvements, each method oversimplifies the data
generating mechanism and there is no gold-standard averaging
approach; \citet{Bro07} and \citet{Goo14} provide thorough discussions on these approaches to MUNE.

\begin{figure}
  \centering
  \includegraphics[width=0.8\textwidth]{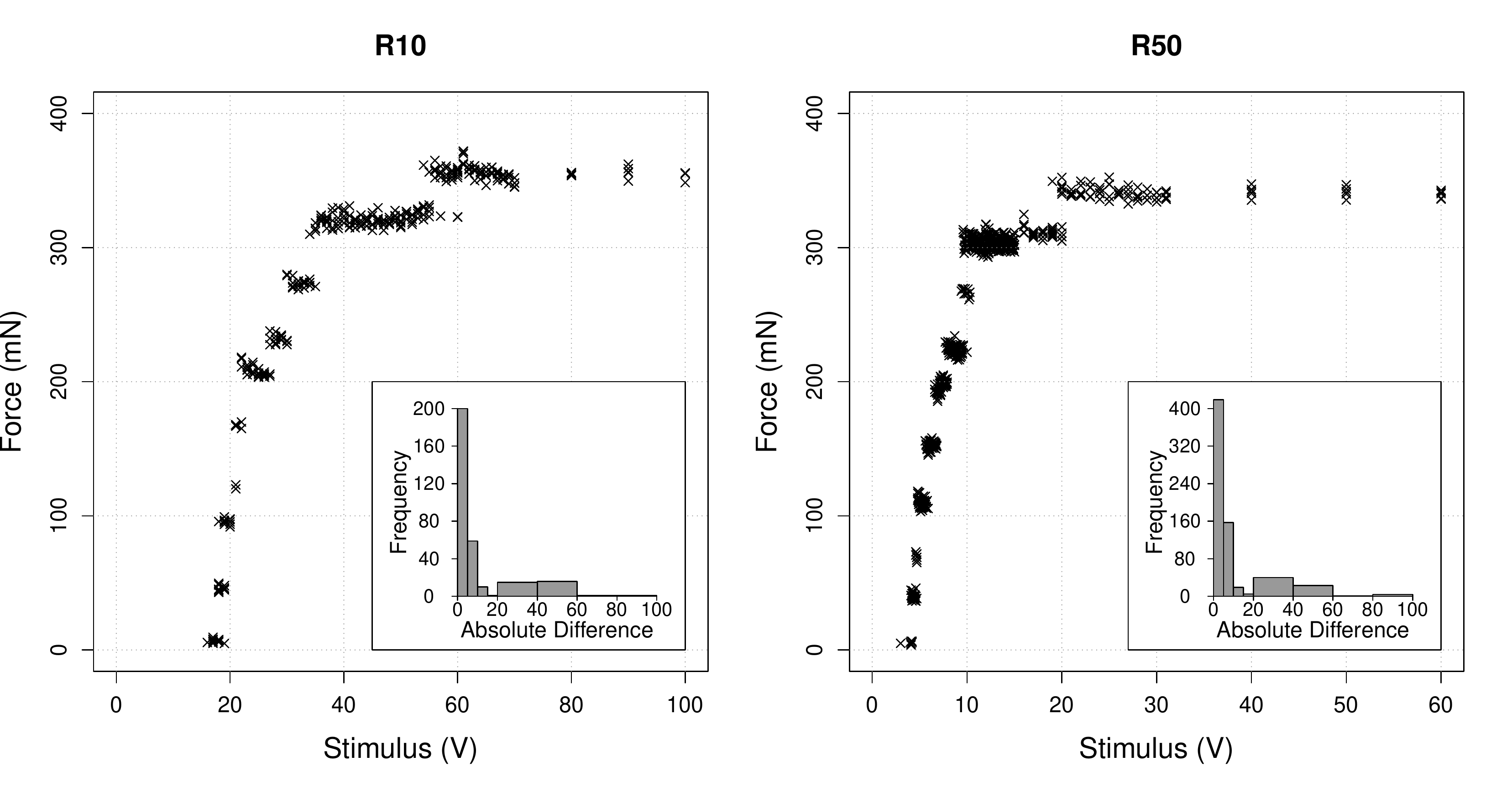}
  \caption{Stimulus-response curve from a rat tibial muscle using 10\,{\micro}sec (left) and 50\,{\micro}sec duration stimuli. Histogram inserts represent the frequency in the absolute difference of twitch forces when ordered by stimulus.}
  \label{fig:RatData}
\end{figure}

 Motor units are more
diverse than simple replicates of the `average' MU, with many factors
influencing their function. A desire for a more complete model for the
data generating mechanism motivated the comprehensive approach
to MUNE in \citet{Rid06}, which proposed three assumptions:
\begin{itemize}
  \item[\textbf{A1}] MUs fire independently of each other and of
    previous stimuli in an all-or-nothing response. Each MU fires
    precisely when the stimulus intensity exceeds a random
    threshold whose distribution is unique to that MU, with a
    sigmoidal cumulative distribution function, called an excitability curve.
  \item[\textbf{A2}] The firing of a MU is characterised by a MUTF which is independent of the
    size of the stimulus that caused it to fire, and has a Gaussian distribution
    with an expectation specific to that MU and a variance common to
    all MUs.
  \item[\textbf{A3}] The measured WMTF is the superposition of the
    MUTFs of those MUs that fired, together with a baseline component
    which has a Gaussian distribution with its own mean and variance.
\end{itemize}
From these assumptions, \citet{Rid06} proposed a set of similar statistical
models each of which assumed a different \textit{fixed} number of
MUs. MUNE thus reduced to selection of a best model, for which the 
 Bayesian information criterion was used.
 The class of methods which performs MUNE within a
Bayesian framework is commonly referred to
as Bayesian MUNE. In a subsequent paper, \citet{Rid07} extended the
method by constructing a reversible jump Markov chain Monte Carlo
(RJMCMC) \citep{Gre95} to sample from the MU-number posterior mass
function directly. However, its implementation is highly challenging
with slow and uncertain convergence particularly when the studied
muscle has many MUs. This is partly attributed to difficulty in
defining efficient and meaningful transitions between models, with
transition rates found to be 0.5--2\% \citep{And07}. The between model
transition rate was improved in \citet{Dro14} where it was noticed
that under Assumption A1, for a given stimulus, the majority of MUs are either almost
certain to fire or almost certain to not fire. Approximating this near
certainty by absolute certainty led to a substantial reduction in the
size of the sample space. The approximate sample space for the firing
events was sufficiently small to permit marginalisation in the
calculation of between-model transition probabilities, increasing the
acceptance rate to 9.2\% with simulated examples. Nevertheless,
substantial issues over convergence remain as the parameter posterior
distributions for models with more than the true number of MUs are multimodal.

In this paper, slight alterations of the neuromuscular assumptions
permit the development of a fully adapted sequential Monte Carlo (SMC)
filter, leading to SMC-MUNE, the first Bayesian MUNE method compatible with real-time analysis.
 As in \citet{Rid06}, the principal inference targets are separate estimates of the marginal likelihood for models with $u=1, \ldots, \Umax$ MUs, for some maximum size $\Umax$. 

The paper proceeds as follows. Section~\ref{sec:Model} presents the
neuromuscular model of \citet{Rid06} for a fixed number of MUs and
defines the priors for the model parameters. Section~\ref{sec:Method}
describes the SMC-MUNE method. Due the complexity of the problem that
MUNE addresses, this section is broken into three parts: inference for
the firing events and associated parameters; inference for the
parameters of the baseline and MUTF processes; and, estimation of the
marginal likelihood so as to evaluate the posterior mass function for
MU-number. Section~\ref{sec:SimStudy} assesses the performance of the
SMC-MUNE method for $200$ simulated data sets. Closer examination of
cases where the point estimate of the  number of MUs was incorrect revealed two
classes of error; an example in each of
these classes is investigated in detail.
Section~\ref{sec:CaseStudy} applies the SMC-MUNE method
to data (collected using the method in \citep{Cas10}) from a rat
tibial muscle that has undergone stem cell
therapy. Section~\ref{sec:Discussion} concludes the paper with a
discussion on the effectiveness of SMC-MUNE and of potential avenues
for improvement.


\section{The neuromuscular model and prior specification}\label{sec:Model}

The three assumptions A1--A3 underpin a comprehensive description of
the neuromuscular system. This section expands on these assumptions to
form the model of the neuromuscular system for a given fixed number
of MUs. Section~\ref{sec:Notation} introduces the notational
convention. Section~\ref{sec:Neuro-model} presents the neuromuscular
model under the assumptions of \citet{Rid06}, and Section~\ref{sec:PriorDist} defines the prior distributions for the model parameters.

\subsection{Notation}\label{sec:Notation}


The total number of MUs operating the muscle of interest is denoted by
$u$ and a particular MU is indexed by $j$.
An EMG data set consists of $T$ measurements whereby the datum for the
$t$th test, $t=1,\ldots,T$, consists of the applied stimulus $s_t$ and
resulting WMTF $y_t$. The data set is re-ordered such that the
observation $y_1, \ldots, y_{\tau-1}$ define baseline measurements
with $s_t=0$ for $t=1, \ldots, \tau-1$, followed by an overall WMTF
$y_\tau$ corresponding to the supramaximal stimulus 
$s_\tau=\max_t(s_t)$ 
where all $u$ MUs are known (by the clinician) to have fired. The
remaining measurements appear in order of increasing stimulus.
The advantages of this
ordering will become evident in Section~\ref{sec:DetailObsProc}.

The reaction of MU $j$ to stimulus $s_t$ is denoted by the indicator
variable $x_{j,t}$, which is $1$ if MU $j$ fires, and hence
contributes to the $y_t$ measurement, and $0$ otherwise. The
$u$-vector of indicators $\bx_t=(x_{1,t},\ldots,x_{u,t})^\top$ defines
the firing vector of the MUs in response to stimulus $s_t$. Given the
experimental set-up, it is assumed that no MUs fire for any baseline
measurement, $x_{j,t}=0$ for each $j=1,\ldots,u$ and $t=1,\dots \tau-1$, and all MUs fire in response to the supramaximal stimulus, $x_{j,\tau}=1$.

A sequentially indexed set of elements, vectors or scalars shall be
represented as $a_{1:t} := \{a_1, \ldots, a_t\}$. The vectors where
all elements are zero or all are unity are denoted by $\zero$ an $\unit$ respectively. The indicator function $\mathbb{I}_{A}(x)$ is $1$ if $x \in A$ and $0$ otherwise.

\subsection{The neuromuscular model}\label{sec:Neuro-model}

Following the assumptions A1--A3 of \citet{Rid06}, the state-space neuromuscular model for the WMTF observations based on a fixed $u$ number of MUs is as follows.
\begin{align}
  X_{j,t} | s_t, \eta_j, \lambda_j & \sim \mathrm{Bern}\left[ F \left(s_t; \eta_j, \lambda_j\right) \right], \label{eq:StateProc}\\
	Y_{j,t} | \mathbf{X}_t = \bx_t, \mub, \nub, \bmu, \nu & \sim \mathrm{N}\left(\mub + \bx_t^\top\bmu, \nub^{-1} + \nu^{-1}\bx_t^\top\unit\right). \label{eq:ObsProc}
\end{align}
The WMTF in \eqref{eq:ObsProc} is the sum of independent Gaussian
contributions, firstly, from a baseline effect of $N(\mub,\nub^{-1})$
and, secondly, from each
MU that fires. If the $j$th MU fires then it makes a
$N(\mu_j,\nu^{-1})$ contribution to the WMTF. The parameters $\bmu =
(\mu_1, \ldots, \mu_u)^\top$, $\nu$, $\mub$, $\nub$ are collectively
referred to as the \emph{observation parameters}. Each firing event in \eqref{eq:StateProc},
$X_{j,t}$, is a Bernoulli random variable with success probability
given by a sigmoidal function $F$ of the stimulus, called the
\emph{excitability curve} \citep{Bro76}. The \emph{excitability parameters}
for the $t$th MU, $\eta_j$
and $\lambda_j$, characterise its excitation features;
conditional on these values, firing events are independent. The
acyclic graph in Figure~\ref{fig:MUNEdag} depicts the dependencies
within the neuromuscular model. Key to the strategy in this paper is
that the observational and excitability parameters are conditionally
independent given the unobserved firing events $\bx_{1:T}$.

\begin{figure}[ht]
  \centering
  \includegraphics[width=0.8\textwidth]{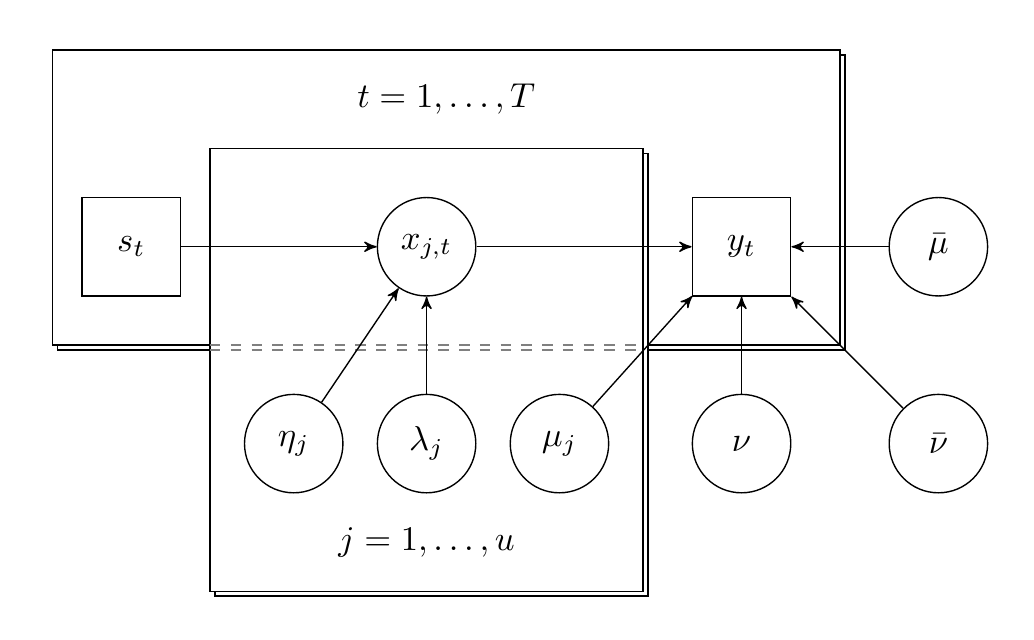}
  \caption{Directed acyclic graph of the neuromuscular model for a
    fixed number of motor units, $u$. Arrows denote direct dependencies between known data (square nodes) and unknown parameters and states (circle nodes). Pallets indicate repeated cases according to the stated index.}
  \label{fig:MUNEdag}
\end{figure}

The excitability curve is a non-decreasing
sigmoid function of the stimulus, parameterised by its median, $\eta$
and the reciprocal gradient at the median: 
$F(s=\eta;\eta,\lambda) = 1/2$, and
$F'(s=\eta;\eta,\lambda) = 1/\lambda$. Under
assumption A1, \citet{Rid06} specifies the excitability curve as the
Gaussian cumulative distribution function (CDF): $F(s) =
\Phi[\delta(s-\eta)]$ where $\Phi(x)$ denotes the standard Gaussian
CDF with $\delta = \sqrt{2\pi}/\lambda$.
Evidence for this definition
\citep{Hal04} focused on the central structure of the excitability
curve by applying a binned chi-squared goodness-of-fit test. However,
evidence to distinguish between this and alternatives such as the logistic
curve will arise, chiefly, from tail events. Moreover, the Gaussian
assumption allows a small, albeit potentially negligible, probability
of a spurious firing event when no stimulus is applied. Given this
contradiction with the experimental design, the following log-logistic form of the excitability is used:
\begin{align}
  F(s; \eta, \lambda) = \left[1 + \left(\frac{s}{\eta}\right)^{-4\eta/\lambda}\right]^{-1}. \label{eq:ECdef}
\end{align}
Nonetheless, the inference method described in Section~\ref{sec:DetailFireProc} is
applicable for any sigmoidal curve.

\subsection{Prior distributions}\label{sec:PriorDist}

The excitability parameters of individual MUs are assumed to be
independent \emph{a priori}. For some upper limits $\etamax$ and $\lambdamax$, the excitability parameters are assigned vague independent beta prior distributions:
\begin{align}
  \frac{\eta}{\etamax} \sim \mathrm{Beta}(1.1,~ 1.1), \quad\quad \frac{\lambda}{\lambdamax} \sim \mathrm{Beta}(1.1,~ 1.1). \label{eq:ECprior}
\end{align}
The shape parameters are chosen so that the densities are
uninformative yet tail off towards the boundaries. The location upper
bound is conservatively set just greater than the supramaximal
stimulus, $\etamax = 1.1s_\tau$. Evidence for specifying the upper bound $\lambdamax$ is
taken from \citet{Hal04} where, for a Gaussian excitability curve, the
coefficient of variation of a random variable whose cumulative
distribution function is given by the excitability curve was
estimated to be 1.65\%. With the log-logistic curve this corresponds to
$\lambda/\eta \approx 3.64\%$. Given that $\eta\le \eta_{max}=1.1s_{\tau}$, we deduce
that $\lambda \le 0.04 s_\tau$. 
The
limitations of the the study of \citet{Hal04}, commented on by
\citet{Maj07}, indicate that a larger bound may be required than
initially suggested, so sensitivity of MUNE to $\lambdamax$ is
investigated in Sections~\ref{sec:SimStudy_under} and \ref{sec:CaseStudy}.

The following pair of four-parameter (multivariate) Gaussian-gamma prior distributions are specified for the observation parameters:
\begin{align}
  \nub & \sim \mathrm{Gam}\left(\ab_0,~ \bb_0\right), \quad\quad \mub|\nub  \sim \mathrm{N}\left(\mb_0,~ \nub^{-1}\cb_0\right),\nonumber\\
  \nu  & \sim \mathrm{Gam}\left(a_0,~ b_0\right),     \quad\quad \bmu|\nu   \sim \mathrm{MVN}_u\left(\M_0,~ \nu^{-1}C_0\right).\label{eq:ObsPrior}
\end{align}
All hyper-parameters are strictly positive scalars except for the
real-valued scalar expectation $\mub_0$, $u$-vector $\M_0$ and $u
\times u$ positive definite matrix $C_0$. The prior distributions
defined for the precision parameters are consistent with
\citet{Rid06}. However, the prior for both baseline and MUTF
expectations differ from the gamma definition of \citet{Rid06}. The
tractability reasons for adopting Gaussian rather than gamma  priors
are detailed in Section~\ref{sec:DetailObsProc}; the problems that
arise from the support now including the whole real line are addressed in Section~\ref{sec:ML}.

The range of MUs to consider, $u=1, \ldots, \Umax$, defines a set of
neuromuscular models. Previous Bayesian MUNE methods defined a uniform
prior on the model space in assuming that each is equally
probable. However, there is typically a preference for identifying the
simplest representation of the underlying process. This is of
particular importance in the presence of alternation where the data
could be equally probably under two or more models. To 
impose an \textit{a priori} preference for smaller models the number
of MUs is given a $\mathsf{Geom}(1/2)$ distribution, truncated at $\Umax$.

\section{Methodology for SMC-MUNE}\label{sec:Method}

The methodology that defines the SMC-MUNE procedure detailed in this
section is based on an approximation to the ideal model defined in
Section~\ref{sec:Model} using, effectively, an approximation to the
prior specification. The reasons for the approximations are twofold:
firstly, the choice of prior is necessary for certain tractable
operations but does not reflect true prior belief; secondly, prior
specification does not lead conveniently and efficiently to sequential
inference yet a simple approximation achieves this goal.   An overview of
the methodology 
 is first
provided, with details about each part given subsequently. Adapting
terminology from sequential inference, the re-ordered index $t$ shall henceforth
be referred to as `time'.

\subsection{Overview}\label{sec:Overview}

The ultimate
aim is to calculate and compare the posterior model probabilities for a
range of models, each with a different number of MUs, $u$. 
Posterior model probabilities are straightforward to obtain once the
marginal likelihood for each model is available. Hence, for a given
model with $u$ MUs, the target for
inference is its marginal likelihood, $f(y_{1:T}|s_{1:T})$; 
throughout this section, for notational simplicity, we suppress the dependence on $u$. This can
be expressed as a product of sequential predictive factors with each defined by:
\begin{align}
  f\left(y_t |~ y_{1:t-1},~ s_{1:t}\right) & = \sum_{\bx_{1:t-1}\in\mathcal{X}_{1:t-1}} f\left(y_t|~ \bx_{1:t-1},~ y_{1:t-1},~ s_{1:t}\right) \mathbb{P}\left(\bx_{1:t-1} |~ y_{1:t-1},~ s_{1:t-1}\right)\label{eq:SeqPredFactor}
\end{align}
where $\mathcal{X}_{1:t} = \{0,1\}^{ut}$ denotes the space for the
sequence of vectors of historical firing events.

The inference scheme is based upon two key observations. Firstly, the
observation and excitability parameters are conditionally independent
given the set of firing events $\bx_{1:T}$. Such an independence
structure separates the observational and firing processes and
simplifies the marginalisation of the parameter space for evaluating
the marginal likelihood. Secondly, conditional on $\bx_t$, the priors
for the observation parameters in \eqref{eq:ObsPrior} are nearly
conjugate for the likelihood in \eqref{eq:ObsProc}.
For a baseline measurement (which has $\bx_t=\zero$) the
posterior has the same form as the prior with tractable 
updates; the same would
be true for a non-baseline measurement ($\bx_t\neq\zero$) if it
were possible to set $\nub^{-1}=0$ and to ignore
the further information on $\SSbar$; such an approximation is
described and justified in Section \ref{sec:DetailObsProc}.

Subject to this approximation, the posterior distribution for the observational
parameters after assimilating $y_{1:t}$ and conditional on $\bx_{1:t}$
is defined by the sufficient statistics and (multivariate)
Gaussian-gamma distributions analogous to the prior specification:
\begin{align}
  \SSbar_t := \left\{\ab_t,~ \bb_t,~ \mb_t,~ \cb_t\right\} \quad \mathrm{and} \quad 
	\SStat_t := \left\{a_t,~ b_t,~ \M_t,~ C_t\right\}, \label{eq:SSdef}
\end{align}
\begin{align}
  \nub|y_{1:t},\bx_{1:t} & \sim \mathrm{Gam}\left(\ab_t,~ \bb_t\right), \quad\quad \mub|\nub,y_{1:t},\bx_{1:t}  \sim \mathrm{N}\left(\mb_t,~ \nub^{-1}\cb_t\right),\nonumber\\
  \nu|y_{1:t},\bx_{1:t} & \sim \mathrm{Gam}\left(a_t,~ b_t\right), \quad\quad \bmu|\nu,y_{1:t},\bx_{1:t}  \sim \mathrm{MVN}_u\left(\M_t,~ \nu^{-1}C_t\right).\label{eq:ObsPriorApx}
\end{align}

Given the assumptions leading to \eqref{eq:ObsPriorApx} the marginal likelihood for the observation $y_t$ conditional on the firing vector $\bx_t$ and sets $\SSbar_{t-1}$ and $\SStat_{t-1}$ has tractable form:
\begin{align}
  f\left(y_t|\bx_t,\SSbar_{t-1},\SStat_{t-1}\right) & = \left\{
  \begin{array}{ll}
    \mathsf{t}\left[y_t;~ \mb_{t-1},~ \frac{\bb_{t-1}}{\ab_{t-1}}\left(\cb_{t-1}+1\right),~
		2\ab_{t-1}\right] \quad & \mathrm{if} ~ \bx_t=\zero,\\
    \mathsf{t}\left[y_t;~ \mb_{t-1}+\bx_t^\top\M_{t-1},~ 
		\frac{b_{t-1}}{a_{t-1}}\left(\bx_{t}^\top C_{t-1}\bx_{t}+\bx_{t}^\top\unit\right),
		~2 a_{t-1}\right] \quad & \mathrm{otherwise}.\\
  \end{array}
  \right.\label{eq:ObsMarginal}
\end{align}
Here, $\mathsf{t}(y; m, v, n)$ denotes the Student's t-density
function on $n$ degrees of freedom with centrality parameter $m$ and
scaling factor $\sqrt{v}$. The statistics $\SSbar_{t-1}$ and
$\SStat_{t-1}$ are deterministic functions of $y_{1:t-1}$ and
$\bx_{1:t-1}$, and are sufficient in that $f(y_t|\bx_{1:t},y_{1:t-1}) \equiv f(y_t|\bx_{t},\SSbar_{t-1},\SStat_{t-1})$.

The 
posterior-predictive mass function for the next excitation vector,
  $\mathbb{P}\left(\bx_t|\bx_{1:t-1},y_{1:t-1},s_{1:t}\right)$\\ 
  $= \mathbb{P}\left(\bx_t|\bx_{1:t-1},s_{1:t}\right)$,
 is given by the following intractable marginalisation:
\begin{align}
\mathbb{P}\left(\bx_t|\bx_{1:t-1},s_{1:t}\right) = \int \mathbb{P}\left(\bx_t|\eta_{1:u},\lambda_{1:u},s_t\right) \pi\left(\eta_{1:u},\lambda_{1:u}| \bx_{1:t-1},s_{1:t-1}\right)~d\eta_{1:u}~d\lambda_{1:u}, \label{eq:ECmarg}
\end{align}
where $\pi\left(\eta_{1:u},\lambda_{1:u}|
\bx_{1:t-1},s_{1:t-1}\right)$ is the posterior for the excitability
parameters given the firing vectors to time $t-1$.
Section~\ref{sec:DetailFireProc} presents a fast numerical quadrature
scheme for evaluating \eqref{eq:ECmarg} to any desired
 accuracy.

The marginalisations over the parameters in \eqref{eq:ObsMarginal} and
\eqref{eq:ECmarg} together provide the predictive:
\begin{align}
  f\left(y_t|~ \bx_{1:t-1},~ y_{1:t-1},~ s_{1:t}\right) & = \sum_{\bx_t\in\mathcal{X}_t} f\left(y_t|~ \bx_{1:t},~ y_{1:t-1},~ s_{1:t}\right) \mathbb{P}\left(\bx_{t} |~ \bx_{1:t-1},~ s_{1:t}\right),\label{eq:Ypred}
\end{align}
Combination of \eqref{eq:Ypred} with the historical firing event mass
function $\mathbb{P}\left(\bx_{1:t}|y_{1:t},s_{1:t}\right)$ would
provide the quantity $f(y_t|y_{1:t-1},~s_{1:t})$ in \eqref{eq:SeqPredFactor}
as required; however,
it is infeasible to track
$\mathbb{P}\left(\bx_{1:t}|y_{1:t},s_{1:t}\right)$ as the dimension of
the event space increases exponentially with time.
Instead, combining 
\eqref{eq:ObsMarginal}, \eqref{eq:ECmarg} and \eqref{eq:Ypred} gives
the conditional mass function for the current firing vector given all
previous firing vectors and all MUTFs to date,
\begin{align}
  \mathbb{P}\left(\bx_{t} |~ y_{1:t},~ \bx_{1:t-1},~ s_{1:t}\right) & = 
	  \frac{f\left(y_t|~\bx_{1:t},~y_{1:t-1},~s_{1:t}\right) \mathbb{P}\left(\bx_t|~\bx_{1:t-1},~s_{1:t}\right)}{f\left(y_t|~\bx_{1:t-1},~y_{1:t-1},~s_{1:t}\right)}. \label{eq:Xupdate}
\end{align}
Expressions in \eqref{eq:Ypred} and \eqref{eq:Xupdate} together 
lead to a fully adaptive sequential Monte Carlo (SMC) sampler which
approximates the historical firing event mass function by the particle
set $\left\{\bx_{1:t}^{(i)}\right\}_{i=1}^{N}$, for a  suitably large
$N$, recursively updating the set for $t=1,\dots,T$.
Algorithm~\ref{tab:Alg}
presents the auxiliary SMC sampler \citep{Pit99} which, given the set of samples drawn from
$\mathbf{X}_{1:t-1}|~y_{1:t-1},~ s_{1:t-1}$,  creates an
unweighted sample from the filtering distribution
$\mathbf{X}_{1:t}|y_{1:t},s_{1:t}$, and approximates
\eqref{eq:SeqPredFactor} via Monte Carlo so as to update the marginal likelihood estimate $\hat{f}(y_{1:t}|~s_{1:t})$.
\begin{algorithm}
\caption{Fully adapted SMC sampler}\label{tab:Alg}
\begin{algorithmic}[1]
 \For{$i$ in $1,\dots,N$} \Comment{Weight}
 \State $\omega_t^{(i)}=f(y_t|\bx_{1:t-1}^{(i)},~y_{1:t},~s_{1:t})$
 \Comment{Using \eqref{eq:Ypred}}
 \EndFor 
 \State $\bar{\omega}_t^{(i)}=\omega_t^{(i)}/\sum_{k} \omega_t^{(k)}$.
 \State Sample auxiliary indices $\{\zeta^{(i)}\}_{i=1}^N$ with
 probabilities $\{\bar{\omega}_{t}^{(i)}\}_{i=1}^N$. \Comment{Resample}
 \For{$i$ in $1,\dots,N$} \Comment{Propagate}
 \State Sample $\bx_t^{(i)}$ with probability
 $\mathbb{P}\left(\bx_t|~y_t,~\bx_{1:t-1}^{(\zeta_i)},~s_{1:t}\right)$. \Comment{Using
 \eqref{eq:Xupdate}}
\State Set $\bx_{1:t}^{(i)} = \left(\bx_{1:t-1}^{(\zeta_i)},~
  \bx_t^{(i)}\right)$. 
\EndFor
\State Set $\log \hat{f}\left(y_{1:t}|~s_{1:t},~u\right) =
\log \hat{f}\left(y_{1:t-1}|~s_{1:t-1},~u\right) - \log N + \log \sum_i
\omega^{(i)}_t$.
\Comment{Update marginal likelihood}
\end{algorithmic}
\end{algorithm}

Although primary interest lies in the marginal-likelihood estimate,
parameter inference is also available to assist in assessing the
quality of fit. The deterministic map to the sufficient statistics from
the set of firing events and responses permits the transformation from
the final particle set $\{X_{1:T}^{(i)}\}_{i=1}^{N}$ to an
$N$-component Gaussian-gamma mixture approximating the posterior
distribution for the observation parameters. A similar transformation
for evaluating the posterior distribution for the excitability
parameters is derived from the approximation to the prior; see
Section~\ref{sec:DetailFireProc} for details.

\subsubsection{Equivalent particle specification and degeneracy}

The Bayesian conjugate structure for the observation process suggests
storing and updating the sufficient statistics when assimilating the
latest observations. Given the prior statistics $\SSbar_0$ and
$\SStat_0$, there is a deterministic map from $(\bx_{1:t-1},~
y_{1:t-1})$ to $(\SSbar_{t-1},~\SStat_{t-1})$. Hence estimates
relating to the observation process at time $t$ are equivalently
expressed with respect to the samples
$\{\SSbar_{t-1}^{(i)},~
\SStat_{t-1}^{(i)},~ \bx_t^{(i)}\}_{i=1}^N$; the storage required for
this set does not increase with with number of observations
assimilated. Since the method relies on these sufficient statistics,
Algorithm \ref{tab:Alg} may be considered as a case of particle
learning \citep{Car10}. For notational clarity, however, the particle
set is described in terms of the historical firing events,
$\bx_{1:t-1}$, unless otherwise specified.

Assimilating the observation, $y_{\tau}$, at the supramaximal
stimulus, $s_{\tau}$, before any of the other non-baseline
observations ensures an update for \emph{all} MUs, $j$, from the initial
vague priors for each $\mu_j$. After this, each $m_j\approx
y_{\tau}/u$, ensuring more sensible predictions in \eqref{eq:ObsMarginal} and hence
\eqref{eq:Ypred}, when a new MU fires. This helps to mitigate against
the inevitable particle degeneracy that occurs with particle
learning. Further mitigation is achieved by iteratively re-running the
algorithm with more and more particles until inferences are stable
(see Appendix \ref{sec:APX}).

\subsection{Details for the firing vector and excitability parameters} \label{sec:DetailFireProc}

At time $t-1$, each particle sample consists of a historical sequence
of firing events for all MUs, and from this an associated joint
posterior for the firing parameters, $\eta_{1:u}$ and $\lambda_{1:u}$,
is derived. A representation of the distribution for the excitability
parameters is sought that is analogous to that described for the
observation parameters in that it should (a) permit simple calculation of the firing event predictive \eqref{eq:ECmarg}, (b) be deterministically updatable when assimilating the current measurement, and (c) provide a concise and sufficient description for the posterior distribution.

From the independence of MU firing under Assumption~A1 and the excitability parameter prior in \eqref{eq:ECprior}, it follows that the predictive for the firing event $\bx_t$ in \eqref{eq:ECmarg} factorises:
\begin{align}
  \mathbb{P}\left(\bx_t|~ \bx_{1:t-1},~ s_{1:t}\right) & = \prod_{j=1}^{u} \iint \mathbb{P}\left(x_{j,t}|~\eta_j,~\lambda_j,~s_t\right) \pi\left(\eta_j,~ \lambda_j|~x_{j,1:t-1},~s_{1:t-1}\right) d\eta_j~d\lambda_j,\label{eq:ECj_marg}
\end{align}
where the posterior at time $t-1$ for the excitability parameters associated with MU $j$ is:
\begin{align}
  \pi\left(\eta_j,~ \lambda_j|~x_{j,1:t-1},~s_{1:t-1}\right) & \propto \prod_{r=1}^{t-1} \mathbb{P}\left(x_{j,r}|~\eta_j,~\lambda_j,~s_r\right) \pi\left(\eta_j\right) \pi\left(\lambda_j\right).\label{eq:ECj_post}
\end{align}
Regardless of the excitability curve definition, this product of firing
probabilities does not lead to a simple conjugate structure with a
concise set of sufficient statistics for the posterior
distribution. Furthermore, whilst for specific values of
$(\eta_j,\lambda_j)$ the update in \eqref{eq:ECj_post} may be
performed sequentially, the integrations for the normalising constant
in \eqref{eq:ECj_post} and the expectation in \eqref{eq:ECj_marg}
require evaluation of the product at arbitrary values in a continuum.
To address these issues, the following approximation is proposed:
\begin{itemize}
  \item[\textbf{B1}] For each MU, store and update at each time point
    a surface proportional to the posterior density
    $\pi(\eta_j,\lambda_j|x_{j,1:t-1},s_{1:t-1})$ at a set of grid of
    points on a regular rectangular lattice $\mathcal{G}$ spanning the
    excitability parameter space. For general $(\eta_j,\lambda_j)$,
    approximate the right-hand side of \eqref{eq:ECj_post} using bilinear interpolation from the four
    nearest grid points.
\end{itemize}
Under this assumption, let $h(\eta,~\lambda)$ be the right-hand side
of \eqref{eq:ECj_post}; then $\tilde{h}(\eta,~\lambda)$, the
interpolated surface specified using points on the unit square in
which $(\eta,\lambda)$ resides, is:
\begin{align*}
  \tilde{h}(\eta,\lambda) & = (1-\eta)(1-\lambda)h(0,0) + (1-\eta)\lambda h(0,1) + \eta(1-\lambda)h(1,0) + \eta\lambda h(1,1), 
\end{align*}
with a similar approximation for
$\mathbb{P}\left(x_{j,t}|\eta_j,\lambda_j,s_t\right) h(\eta_j,
\lambda_j)$ based on interpolating this between grid points.
The resulting approximations for the normalising constant in
\eqref{eq:ECj_post} and the integral in \eqref{eq:ECj_marg},
therefore, correspond to iterative (over the two dimensions) 
application of the compound trapezium rule.  This approach provides a deterministic updating procedure for maintaining the excitability posterior density up to a constant of proportionality for each point on the regular lattice.


A na\"ive implementation of the above scheme would evaluate the
posterior density for each grid point, MU and particle
sample. However, two posterior densities will only differ if the
historical firing events differ. Consider any two particles, $i$ and
$i'$, each with an associated MU, $j$ and $j'$ respectively, that have
identical firing histories: $x_{j,1:t}^{(i)} =
x_{j',1:t}^{(i')}$. Since the priors for the excitability parameters
are identical for all MUs then the posterior distribution for these
two MUs on these two particles are identical. Efficiency gains are
therefore achieved by storing a single grid of values for each unique
firing pattern to date. 

A higher-order Newton-Cotes numerical integration method would produce
a more accurate estimate of \eqref{eq:ECj_marg}, but the associated
interpolated density surface of piecewise polynomials would not be
guaranteed to be bounded below by zero, making an inspection of parameter estimates for assessing model fit problematic. Alternatively, quadrature on adaptive sparse grids \citep{Bun03}, where the grid is finer at regions of high curvature, could improve estimator accuracy over the static regular rectangular lattice. However, this would be achieved at the expense of additional implementation complexity and further approximation error when estimating the surface at infilled lattice points.

\subsection{Details concerning the observation
  process} \label{sec:DetailObsProc}
Consider the observation model \eqref{eq:ObsProc}.
At time $t\le \tau-1$, when no MUs fire, $\bx_t = \zero$, the observation, $y_t$,
provides no new information about the observation parameters for the
MUs, $\SStat_t = \SStat_{t-1}$, and
$Y_{j,t}|\bx_t=0,\mub,\nub,\bmu,\nu\sim \mathrm{N}(\mub,\nub^{-1})$.
Standard conjugate updates may, therefore, be applied to obtain $\SSbar_t$ as follows:
\begin{align*}
  \mb_t = \mb_{t-1} + \frac{y_t-\mb_{t-1}}{1+\cb_{t-1}},\quad
	\cb_t = \frac{\cb_{t-1}}{1+\cb_{t-1}},\quad
	\ab_t = \ab_{t-1} + \frac{1}{2},\quad
	\bb_t = \bb_t + \frac{(y_t-\mb_{t-1})^2}{2(1+\cb_{t-1})}. 
\end{align*}

When $t\ge \tau$, at least one MU fires and tractable updates are not
possible. However, in real experiments, because of the precautions
detailed in Section \ref{sec:Intro},
the variance (and expectation) of the baseline noise are
generally much smaller than the variability in response from a given
MU when it fires. For example
\cite{Hen06} find a ratio of an order of magnitude. We, therefore
make the following approximation:

\begin{itemize}
  \item[\textbf{B2}] When assimilating a non-baseline observation,
  $\SSbar$ is kept fixed at its previous value, and for updating
    $\SStat$ it is assumed that $\nub^{-1}=0$.
\end{itemize}

Approximation B2 implies that for $\bx_t\neq\zero$,
\begin{align*}
  Y_t|~ \nub,~ \bmu,~ \nu,~ \SSbar_{t-1},~ \SStat_{t-1} & \stackrel{apx.}{\sim} 
	  \mathrm{N} \left(~\mb_{t-1} + \bx_t^\top\bmu,~ \nu^{-1}\bx_t^\top\unit~ \right), 
\end{align*}
which, given distributions at time $t-1$ as specified in
\eqref{eq:ObsPriorApx}, leads to the desired tractable updates for the sufficient
statistics for \eqref{eq:SSdef} as follows: $\SSbar_t = \SSbar_{t-1}$ and 
\begin{align*}
  \M_t &= \M_{t-1} + q_t C_{t-1}\bx_t(y_t - \mb_{t} - \bx_t^\top\M_{t-1}) &
  C_t  &= C_{t-1} - q_t C_{t-1} \bx_t \bx_t^\top C_{t-1} \nonumber\\
  a_t  &= a_{t-1} + \frac{1}{2} &
  b_t  &= b_{t-1} + \frac{q_t}{2}(y_t - \mb_{t} - \bx_t^\top\M_{t-1})^2, 
\end{align*}
where $q_t=(\bx_t^\top\unit + \bx_t^\top C_{t-1} \bx_t)^{-1}$. In essence, the approximate observation
process decouples the learning about the observational parameters:
when no MU fires then $(\mub,\nub)$ is updated, else $(\mu,\nu)$ is updated.

After assimilating the baseline observations 
$y_1,\dots,y_{\tau-1}$, both $\nub^{-1}$ and $\mub$ are known (and known to
be small) with considerable certainty. Thus, approximating $\nub^{-1}$ as
$0$ and considering $\mub$ to be a point mass at $\mb$ is
reasonable. Furthermore, the prior for $\nu$ does not need to be set
until just before the observation $y_{\tau}$ is assimilated. Given the
tight posterior for $\nub$ at this juncture it is,
therefore, possible to incorporate the knowledge that $\nub>>\nu$ into
the vague prior for $\nu$ (which is conceptually equivalent to specifying an initial
joint prior on $\nub$ and $\nu$).
Letting $\nub_{\tau-1}^{\mbox{\scriptsize med}}$ denote the posterior
median of $\nub$ at time $\tau-1$,
tuning parameters 
$\epsilon<<1$ and $\delta<<1$ are chosen such that 
 $\mathbb{P}(\nu>\epsilon\nub)\approx \delta$ is desired. Given that
\begin{align}
  1-\delta = \mathbb{P}(\nu\leq\epsilon\nub) \approx \mathbb{P}(\nu\leq\epsilon\nub_{\tau-1}^{\mbox{\scriptsize med}}) = \mbox{\textsf{Gam}}(b_{\tau-1}\epsilon\nub_{\tau-1}^{\mbox{\scriptsize med}};~a_{\tau-1}),\label{eq:DefineNuPrior}
\end{align}
where $\mbox{\textsf{Gam}}(x; \alpha)$ is the cumulative distribution
function evaluated at $x$ of a gamma random variable with shape $\alpha$ and unit rate, a practical specification for the prior for $\nu$ is obtained by defining a small $a_{\tau-1}=a_0$ and then solving \eqref{eq:DefineNuPrior} for $b_{\tau-1}$.

\subsection{Improving the marginal likelihood estimate} \label{sec:ML}

The following post-processing development is motivated by the analysis
of a particular simulated dataset where the point estimate for the
number of MUs is one greater than the true number. The detailed
analysis in Section \ref{sec:SimStudyOver} shows that the extra MU has
a very weak expected 
MUTF and that it, effectively, acts simply to increase the variability
in the response.
The problem arises because the 
$u$-vector, $\bmu$, of expected MUTF contributions has a Gaussian
prior which, to allow reasonable uncertainty across the typical range
of believable MUTF contributions, also places a non-negligible
prior mass at low and even negative values. Negative expectations for
an individual MU need not be prohibited by the data provided
that MU is always inferred to  fire alongside another MU with a
 positive expectation of similar or larger magnitude. The fact that the parameter
suport permits this possibility potentially increases the marginal
likelihood for a model which is larger than that necessary to explain the
data.

Guaranteeing
positive MUTFs greater than some minimum \citep{Bro03, Maj07} would
require a change to the likelihood. The approach taken in
\citet{Rid06} is to specify 
independent left-truncated gamma prior distributions for the the expected MUTFs,
 $\mu_j$ for $j=1,\ldots,u$. However any such change would not lead to
the tractable updates required for the concise sequential analysis
 described in 
 Sections \ref{sec:Overview} and \ref{sec:DetailObsProc}. 
Within the constraints of the algorithm overviewed in Section
\ref{sec:Overview}, the natural mechanism for preventing these
undesirable scenarios is via post-processing: the conditional prior for
$\bmu|\nu$ in \eqref{eq:ObsPrior} is re-calibrated by truncating it to the region
$M=[\mumin,\infty)^u$ for some minimum MUTF $\mumin$. It follows that the re-calibrated marginal prior for $\bmu$ is:
\begin{align}
  \tilde{\pi}(\bmu) & = \frac{1}{\pi(M)}\pi(\bmu)\mathbb{I}_M(\bmu),\label{eq:TruncMU}
\end{align}
where $\pi(\bmu)$ is the multivariate Student's t-density centred at $\M_0$ with shape matrix $\frac{b_0}{a_0}C_0$ and $2a_0$ degrees of freedom, and, with a slight abuse of notation, $\pi(M)=\int_M\pi(\bmu)d\bmu$. The effect on the marginal likelihood from the prior re-calibration is examined by Proposition~\ref{prop:AdjML}.

\begin{proposition}\label{prop:AdjML}
Let $f(y_{1:T}|s_{1:T})$ denote the marginal likelihood defined in
Section~\ref{sec:Overview}. The
re-calibrated marginal likelihood, denoted by
$\tilde{f}(y_{1:T}|s_{1:T})$, resulting from truncating 
 the prior for $\bmu$ in \eqref{eq:ObsPrior} to $M$ is:
 \begin{align}
  \tilde{f}(y_{1:T}|s_{1:T}) = \frac{\pi(M|y_{1:T},s_{1:T})}{\pi(M)} f(y_{1:T}|s_{1:T}) = \frac{\pi(M|y_{1:T},s_{1:T})}{\pi(M|y_{1:\tau-1},s_{1:\tau-1})} f(y_{1:T}|s_{1:T}), \label{eq:recalib}
\end{align}
where $\pi(M|y_{1:t},s_{1:t}) = \int_M \pi(\bmu|y_{1:t},s_{1:t}) d\bmu$.
\end{proposition}

\begin{proof}
Expressing the re-calibrated marginal likelihood as a marginalisation of $\bmu$ and substituting the definition \eqref{eq:TruncMU} produces the first equality by:
\begin{align*}
  \tilde{f}(y_{1:T}|s_{1:T}) =\int \tilde{\pi}(\bmu) f(y_{1:T}|\bmu,s_{1:T}) d\bmu
	 = \frac{\int_M \pi(\bmu)f(y_{1:T}|\bmu,s_{1:T}) d\bmu}{\pi(M)}
	 = \frac{\pi(M|y_{1:T},s_{1:T})}{\pi(M)} f(y_{1:T}|s_{1:T}).
\end{align*}
The second equality in \eqref{eq:recalib} arises as the first $\tau-1$ observations relate exclusively to the baseline.

Evaluation of $\pi(M|y_{1:\tau-1},s_{1:\tau-1})$ is straightforward
since it is an orthant probability for the multivariate Student-t
distribution. In contrast, the posterior probability is estimated from
the $N$-component-mixture approximation of the posterior distribution by the final particle set:
\begin{align*}
  \hat{\pi}(M|y_{1:T},s_{1:T}) & = \frac{1}{N} \sum_{i=1}^{N} \pi(M|\bx_{1:T}^{(i)}, y_{1:T}). 
\end{align*}
\end{proof}

There is no theoretical argument against assuming \eqref{eq:TruncMU}
from the outset. Indeed, the firing events sampled in the propagation
step of Algorithm~\ref{tab:Alg} would then account for the restriction to
the $\bmu$ parameter space and therefore directing particle samples to
a more appropriate approximation for posterior parameter
estimates. However, implementing this scheme requires at most $N2^u$
orthant evaluations of the multivariate Student's t-distribution per
time step in calculating the re-sampling weights. Standard procedures
for evaluating these orthant probabilities \citep{Genz09} are 
expensive, so the computational time of the resulting
SMC-MUNE algorithm would increase substantially.

\section{Simulation study}\label{sec:SimStudy}

The performance of the SMC-MUNE algorithm is now assessed using
$200$ simulated data sets, $20$ for each true number of MUs of
$u^*=1,\ldots,10$. Each data set consists of $T=220$ measurements with
$\tau=21$ so that the first $20$ observations correspond to the
baseline, $s_t=0$\,V, and these are followed by the supramaximal stimulus $s_{21}=40$\,V. 

All MUs are excited according to the log-logistic curve
\eqref{eq:ECdef} with MU parameters simulated anew for each dataset as
follows: $\eta_j\sim \mathsf{Unif}(5,40)$,
$\lambda_j\sim\mathsf{Gamma}(2,8)\mathbb{I}(\lambda_j<10)$, $\mu_j\sim \mathsf{N}(40,20^2)\mathbb{I}(\mu_j>20)$, $\nu^{-1}\sim \mathsf{Unif}(1,5)$. 
The measurement units for excitation parameters are all in V and the expected MUTFs 
are in mN with variance parameter in mN${}^2$. 
Parameters were independent except
for the following constraints, where $(j)$ is the index of the MU with the $j$th
highest $\eta$ value: $\eta_{(j)}-\eta_{(j-1)}>2$ (neigbours must
be separate) and $|\mu_{(j)}-\mu_{(j-1)}|>4$ (neighbours
should have distinct expectations). To test the ability of SMC-MUNE, these values ensure a greater
range in MUTF contributions and in the excitation parameters than is
typically observed in practice. For example,
 the average 
coefficient of variation is $11.2\%$ (as opposed to $1.65\%$; see Section \ref{sec:PriorDist}). This resulted in $79\%$
of datasets containing at least one alternation event. 
 Additional noise was generated
as in \eqref{eq:ObsProc} with $\mub=0$\,mN and $\nub^{-1}=0.25^2$\,mN${}^2$.

To each data set, a set of neuromuscular models was fitted with a
number of MUs, $u$, ranging from $1$ up to a maximum size of
$\Umax=12$. The sufficient statistics for the parameter prior
distributions are provided in Appendix~\ref{sec:APX}.
To control 
the Monte Carlo variability and the error in the numerical
integration,
the particle set
size, $N$, and the lattice size for numerical integration were
iteratively increased until estimates of the marginal likelihood were
stable; see Appendix
\ref{sec:APX} for further details. The point estimate of the number of MUs is taken to be the maximum \textit{a posteriori} (MAP) model, $\hat{u}$, and uncertainty in the estimate is quantified by the 95\% highest posterior credible set (HPCS); the minimal set of models where their total posterior probability is at least 95\%. In addition, the estimated posterior probability for the true model, $p_{u^*} = \hat{\mathbb{P}}(U=u^*|y_{1:T},s_{1:T})$, is evaluated. 

\begin{table}
\centering
\caption{Summary of the MU-number posterior mass functions and required numerical resource for 200 simulated data sets.}
\label{tab:SimSummary}
\begin{tabular}{lcccccc}
\hline
Number of MUs, $u^*$									& $\leq 5$ 	& 6 	& 7 		& 8 		& 9 		& 10		\\
\hline
No. where $\hat{u}=u^*$								& 100			& 19		& 19 		& 16  	& 15  	& 12		\\
No. where $u^*$ in HPCS								& 100			& 20		& 20		& 19		& 19		& 20		\\
Avg. size of HPCS											& 1.11		& 1.70	& 2.10	& 2.05	& 2.35 	& 2.45 	\\
Avg. $\hat{p}_{u^*}$ (\%)							& 97.95		& 89.20	& 80.45	& 69.42	& 62.70	& 54.68	\\
Avg. particle set size for $u^*$			& 5000 		& 5250 	& 6000 	& 7500 	& 7500 	& 8250 	\\
Avg. particle set size for $\hat{u}$	& 5000 		& 5000 	& 6000 	& 7500 	& 9250 	& 7750 	\\
Avg. $n{\times}n$ lattice size for $u^*$ 					& 30.0 		& 30.5 	& 30.5 	& 32.0 	& 32.5 	& 32.0 	\\
Avg. $n{\times}n$ lattice size for $\hat{u}$ 			& 30.0		& 30.0	& 30.5 	& 32.0	& 35.0	& 32.0 	\\
\hline
\end{tabular}
\end{table}


Table~\ref{tab:SimSummary} presents summaries of the
mass functions of the number of MUs and descriptions of the resource
required as functions of the true number of MUs.
The MAP estimate
corresponded to the true number of MUs for all data sets generated
from five or fewer MUs, and for most of these datasets the HPCS contained
only the true model.
For true sizes of greater than five the MAP estimate was correct for
$81$ of the $100$ data sets and the HPCS contained the truth for all
but two data sets.

It is
unsurprising that the uncertainty in the MU-number increases with the
true number of MUs; this is visible both as an increase in the average
size of the HPCS and a reduction in the average posterior probability
for the true number.
In addition, both the number of particles required to control Monte Carlo
variability and the size of the numerical lattice required for accurate
numerical integration also increase as with the true number of MUs.
 This demonstrates the challenge of MUNE for large
neuromuscular systems that possess complex features resulting from
alternation.

Of the $19$ datasets where the MAP estimate $\hat{u}$ did not
correspond to the truth, $u^*$, one dataset had $\hat{u}>u^*$ with the
rest (including the two outliers) having $\hat{u}<u^*$. 
 The stimulus-response curves for the first case and a typical example
 of the second are presented in Figure~\ref{fig:SimCaseData} and are
 discussed in turn below.

\begin{figure}
  \centering
  \includegraphics[width=0.8\textwidth]{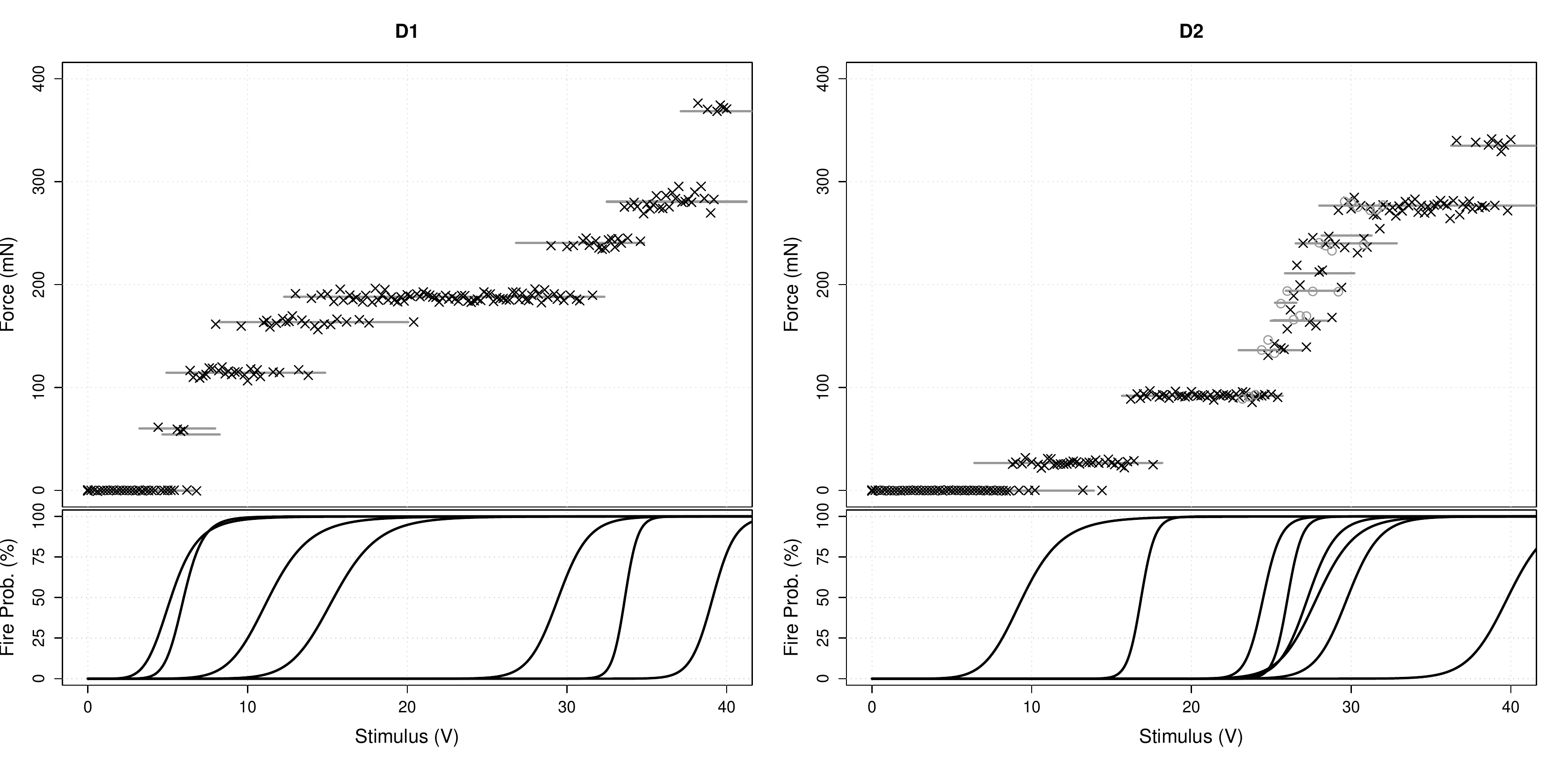}
  \caption{Stimulus-response curve (top) for the simulated data with lines representing the expected WMTF over the stimuli intervals where the joint firing probability is greater than 5\%  according to the individual excitability curves (bottom). Left: Data set D1 contains $u^*=7$ MUs but $\hat{u}=8$. Right: Data set D2 contains $u^*=8$ MUs but $\hat{u}=7$. Circle points, additional 23 simulations over the 23--32\,V alternation period involving 5 MUs.}
  \label{fig:SimCaseData}
\end{figure}

\subsection{Over-estimation} \label{sec:SimStudyOver} 

The first data set, D1, contains $u^*=7$ MUs in truth but the SMC-MUNE
method produces a MAP estimate of $\hat{u}=8$. The posterior
probability of the true model is $\hat{p}_{u^*}=14.9\%$ and this
model, along with the larger 9 MU model,
is contained with a 95\% HPCS. Parameter estimates for the MAP model
(Table~\ref{tab:D1muEst}) show that the penultimate MU has a median
expected twitch force of 9.6\,mN with a credible upper bound of
15.7\,mN, much lower than the 20\,mN simulation threshold. Figure~\ref{fig:PredDen_S37}
presents the the construction of the predictive WMTF density for the
true and MAP models at a 37\,V stimulus. The local modes in the model
containing the true MU-number correspond uniquely to particular firing
combinations. In contrast, the weak MU in the MAP model principally serves to increase the variability around a specific WMTF response level rather than describing a distinct MU.

In light of these concerns, the marginal likelihood estimates are
adjusted according to Section~\ref{sec:ML} with a conservative lower bound of
$\mumin=15$\,mN to guard against small MUs that, when firing, are
indistinguishable from other combinations. The corrected posterior
mass function places 89.3\% of the mass on the correct, seven-MU model, with
10.7\% mass on the eight-MU model. The estimates of expected MUTF in
Table~\ref{tab:D1muEst} for the seven-MU model are similar to those
prior to the adjustment and are still close to the true values from
which the data was generated. However, the prior adjustment for the eight-MU hypothesis has a significant effect on the penultimate MU and, so
as to preserve the overall maximum WMTF, a small reduction in the estimated $\mu$s for its neighboring MUs.

\begin{table}%
  \centering
  \caption{Expected MUTF median and 95\% credible interval estimates for MUs with high excitation threshold from the true ($u^*=7$) and MAP ($\hat{u}=8$) models, with and without post-process truncation ($\mumin=15$\,mN) on data set D1.}
  \label{tab:D1muEst}
  \begin{tabular}{lcccc}
    \hline
    Parameter				& $\mu_6$ 			& $\mu_7$ 			& $\mu_8$ 			& $\nu^{-1}$				\\
    True					& 40.2\,mN          & 87.9\,mN          & --                & 4.54
    \,mN${}^{2}$	\\
    \hline
    $u=7$					& 40.5 (37.7, 43.5) & 91.2 (86.6, 95.9) & --                & 3.85 (3.13, 4.81)	\\
    $u=8$					& 36.3 (32.4, 40.2) & 9.6 (4.7, 15.7)   & 86.7 (80.3, 92.7) & 3.22 (2.57, 4.18)	\\
    $u=7$ \& $\mumin=15$	& 40.5 (37.8, 40.2) & 91.3 (86.8, 95.7) & --                & 3.90 (3.14, 4.92)	\\
    $u=8$ \& $\mumin=15$	& 35.7 (31.1, 40.0) & 15.7 (15.0, 20.7) & 80.4 (71.4, 86.6) & 3.23 (2.60, 4.09)	\\
    \hline
  \end{tabular}
\end{table}

\begin{figure}
  \centering
  \includegraphics[width=0.8\textwidth]{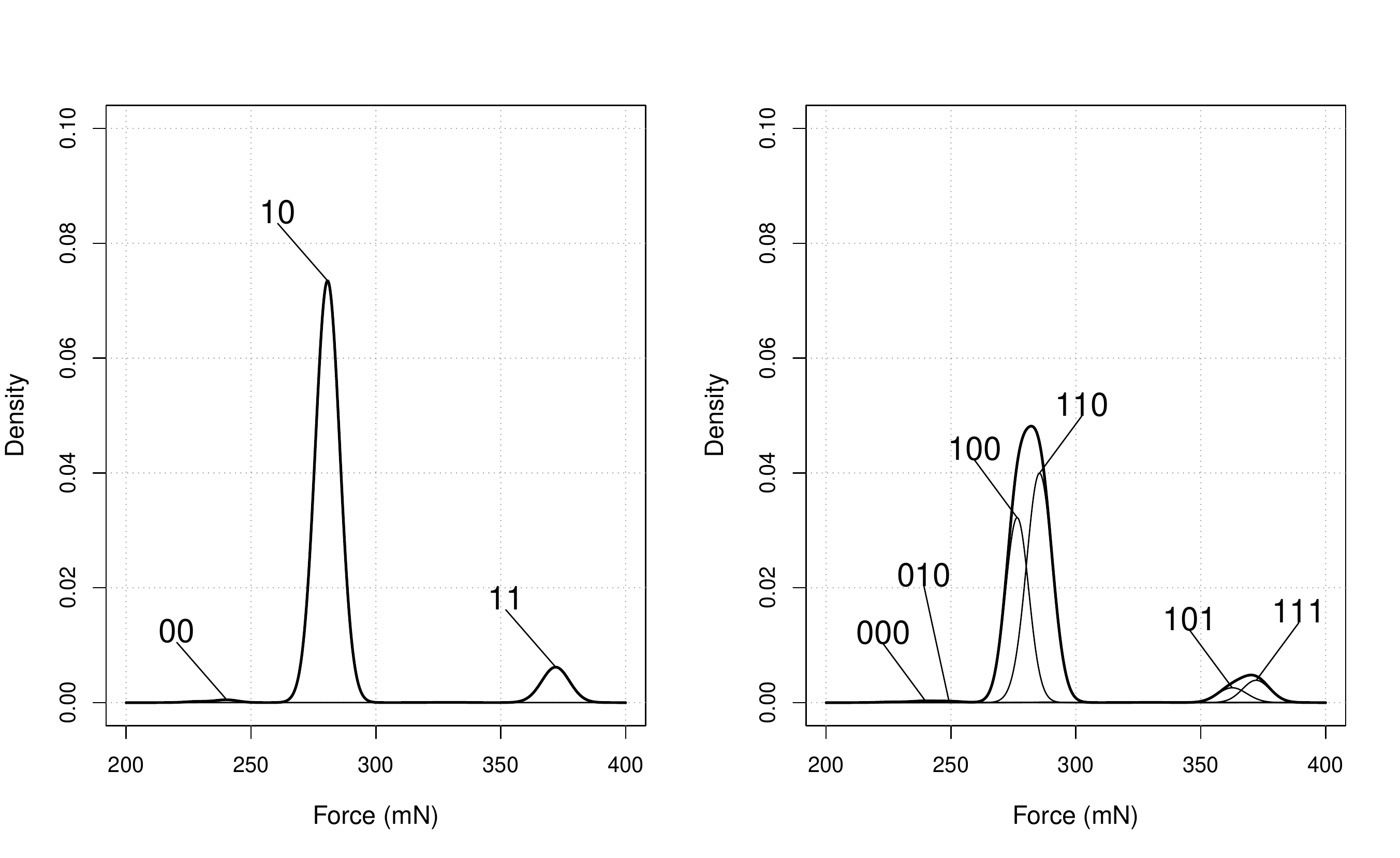}
  \caption{Predictive density (thick line) at stimulus 37\,V from the
    seven (left) and eight (right) MU model without post-process
    adjustment. Thin lines identify the contribution to the predictive
    for the indicated firing combinations associated to the final few
    MUs. In both cases, the first five MUs fire with near
    certainty. Most firing combinations with negligible 
    predictive probabilities are omitted from the plot.}
  \label{fig:PredDen_S37}
\end{figure}

\subsection{Under-estimation} \label{sec:SimStudy_under} 

The second data set, D2, contains $u^*=8$ MUs and presents a period of
alternation between 23--32\,V which involves five MUs. The SMC-MUNE
procedure, however, estimates $\hat{u}=7$ and gives this a high posterior
probability of 97.1\% after applying the post-process adjustment at
$\mumin=15$\,mN. The main source for this under-estimation arises
through the over-estimation of the excitability scale parameter
(Table~\ref{tab:lamMAX}) for the fourth MU ($\lambda_4$), so that the stimulus
interval for probabilistic firing behavior is nearly three times wider
than it should be. Consequently, this incorrectly estimated MU acts
as a surrogate for MU-number $6$, which has similar twitch force properties.

One potential solution is to reduce the upper bound for the scale
parameter $\lambdamax$ in \eqref{eq:ECprior} to constrain estimation
against shallow excitability curves. Table~\ref{tab:lamMAX} presents
scale parameter estimates for selected MUs at the original
($\lambdamax=14$V) and reduced ($\lambdamax=7$V) upper bounds. Under the
reduced bound the 8 MU model becomes a member of the HPCS, but the MAP
estimate remains at $\hat{u}=7$ with a high posterior probability of
94.6\%. Although a further reduction to $\lambdamax$ might be
appealing, this action is likely to be detrimental in determining good
model fits. For example, the scale parameter of the first MU, which
has a true value of 5.0\,V is accurately estimated whether
$\lambdamax$ is $7$ or $14$,
principally because its excitability curve is well separated
from the other curves, but a further reduction in $\lambda_{\max}$
risks the introduction of an additional, spurious MU to explain the
low-stimulus observations. 

\begin{table}[ht]
\centering
\caption{Motor unit posterior probabilities and excitation parameter estimates for selected MUs from D2 with scale upper bound at $\lambdamax=14$\,V and $\lambdamax=7$\,V, and with 23 additional measurements at $\lambdamax=14$\,V.}
\label{tab:lamMAX}
\begin{tabular}{lcP{18mm}P{18mm}P{18mm}P{18mm}P{18mm}P{18mm}}
  \hline 
 					& True 	& \multicolumn{2}{c}{$\lambdamax=14$}	& \multicolumn{2}{c}{$\lambdamax=7$}	& \multicolumn{2}{c}{Extra}				\\ 
  $u$				& 8		& 7					& 8					& 7					& 8					& 7 				& 8					\\
  \hline
  $\mathbb{P}(u|y)$	& --	& 96.7\%			& 3.3\%			& 94.6\%			& 5.4\%			& 28.7\%			& 71.3\%			\\
  $\eta_4$ 			& 26.0\,V 	& 26.9 \newline (25.4, 28.8) & 26.6 \newline (25.3, 28.6) & 27.0 \newline (25.4, 28.9) & 26.5 \newline (25.2, 28.5) & 26.7 \newline (25.7, 27.7) & 26.4 \newline (25.6, 27.6) \\ 
  $\eta_5$			& 27.3\,V 	& 27.8 \newline (26.0, 29.1) & 27.4 \newline (25.5, 28.9) & 27.8 \newline (25.9, 29.0) & 27.5 \newline (25.6, 28.9) & 27.4 \newline (26.2, 28.5) & 27.2 \newline (25.9, 28.3) \\ 
  $\eta_6$ 			& 27.9\,V 	& -- 				& 27.9 \newline (26.5, 29.5) & -- 				& 27.9 \newline (26.6, 29.5) & --			 	& 27.5 \newline (26.5, 28.8) \\ 
  $\lambda_4$ 		& 1.8\,V 	& 4.5 \newline (1.8, 7.6) 	& 3.6 \newline (1.0, 7.6) 	& 4.3 \newline (1.9, 6.6) 	& 3.1 \newline (0.7, 6.3) 	& 4.0 \newline (1.8, 7.8) 	& 2.5 \newline (0.9, 6.3) 	\\ 
  $\lambda_5$ 		& 3.6\,V 	& 4.1 \newline (2.2, 7.3) 	& 4.7 \newline (1.8, 7.9) 	& 4.0 \newline (2.1, 6.5) 	& 4.4 \newline (1.7, 6.6) 	& 3.7 \newline (2.1, 6.3) 	& 4.6 \newline (1.8, 8.3) 	\\ 
  $\lambda_6$		& 4.8\,V 	& -- 				& 4.7 \newline (1.6, 8.1) 	& -- 				& 4.4 \newline (1.4, 6.6) 	& -- 				& 4.4 \newline (2.3, 7.5) 	\\ 
   \hline
\end{tabular}
\end{table}

In the original analysis, $\lambda_4$ 
is mis-estimated because of the limited information available in the
observations to adequately describe the period of alternation between
23--32\,V which involves five MUs. To show that this is the case, an
additional 23 observations were generated evenly over this interval;
see Figure~\ref{fig:SimCaseData}. This modest addendum to the data set
is sufficient for the true model to be identified, $\hat{u}=8$, and
with a posterior probability of 71.3\%, and with better scale
parameter estimates. However, the increase in computational resource required to
obtain the same degree of Monte Carlo and numerical accuracy was substantial: from 5000 to 25000 particles and from a $30{\times}30$ to $50{\times}50$ lattice for the eight-MU hypothesis.

\section{Case study: rat tibial muscle}\label{sec:CaseStudy}

The case study arises from \citep{Cas10} where a rat tibial muscle
(medial gastrocnemius) receives stem cell therapy to encourage
neuromuscular activation after simulating paralysis. The two data
sets, presented in Figure~\ref{fig:RatData}, are generated by
applying stimuli for different durations. The first data set, using
10{\micro}sec duration stimuli, consist of $T=304$ observations,
including 11 baseline measurements and a maximal stimulus of
100\,V. In contrast, the second data set was collected using
50{\micro}sec duration stimuli and consists of $T=669$ observations,
including 7 baseline measurements, and with a maximal stimulus of
60\,V. The data sets are named R10 and R50 respectively. Since both data sets are collected from the same neuromuscular system it is expected that MUNE should be similar between the data sets.

Na\"ive assessment of the stimulus-response curves by counting the
number of distinct levels of twitch force suggests that there are
perhaps nine or ten MUs, but this would ignore any potential features
arising from alternation. The histogram inserts in
Figure~\ref{fig:RatData} present frequency in absolute difference
between consecutive twitch forces when ordered by stimulus
intensity. The highest frequency occurs at low differences and represents the within-MUTF variability whereas the less-frequent, larger differences appear due to the firing of different combinations of MUs. In both cases a minimum expected MUTF of $\mumin=15$\,mN is suitable to correct against the estimation of MUs with negligible contribution to the observed twitch forces.

The SMC-MUNE procedure was applied up to a maximum model size of
$\Umax=12$ with prior sufficient statistics and algorithmic parameters
as specified in Appendix~\ref{sec:APX}. For both data sets, the estimated motor unit
number posterior mass function (Table~\ref{tab:RealResults})
identifies the MAP estimate as $\hat{u}=8$,  with this being the only
member of the HPCSs. There is a noticeable difference in the
computational resources required as the MAP model for R50 required
twenty times more particles and three times finer lattice than that
for data set R10. This is in part reflective of the relative sizes of
the data sets, but may also relate to the relative complexities of 
the state-spaces for the firing vectors.

Figure~\ref{fig:RealParam} presents the estimated excitability curves
for each of the MUs, with MUs labelled in order of increasing $\mathbb{E}[\eta|y_{1:T}]$.
 First, the location parameters under the 50{\micro}sec duration stimuli are approximately four times lower than the corresponding parameters under 10{\micro}sec duration stimuli. This difference in scale corresponds to Weiss's law \citep{Bos83} that relates the excitation of the neuron to the charge built-up in the cell. Despite this, it is clear that the majority of the MUs are excited within a short stimulus window with only the last MU requiring a larger stimulus to be excited. This high degree of activity at low stimulus is reflective of the sudden early rise in the stimulus-response curves.

To compare MUs between data sets, the coefficient of variation for the
random variable associated with each excitability curve is presented
in Figure~\ref{fig:RealParam}c. Apart from the first MU, the 95\%
credible intervals from each data set for a given MU overlap,
suggesting similar coefficients of variation for the MUs; 
this might be  anticipated since
measurements are taken from the same neuromuscular system. These
similar estimated coefficients are larger than the estimate in   
 \citet{Hal04}, which is presented for comparison. This reflects the experimentation where the developed neurons are less stable and are yet to restore full and healthy motor function.

Table~\ref{tab:Xest} present the most probable firing combinations for each visibly distinct response level in Figure~\ref{fig:RatData}. The estimated firing behavior of each MU, after label-swapping similarly excitable MUs for R50, are very similar between the two data sets. It can then be suggested that the level at approximately 120\,mN in both data sets and at about 70\,mN in R50 are potential consequences of alternation as MUs that fired in contributing to weaker WMTFs are latent in forming these response levels.
However, a difference in estimated firing behavior occurs at the
120\,mN response level whereby the SMC-MUNE procedure obtained two
different model fits; MU1+MU4 in R10 and MU2+MU3 in R50. As a
consequence, the estimated excitation range for MU1 in R50 is
unusually large, leading to a relatively flat excitability curve with
an enlarged coefficient of variation (Figure \ref{fig:RealParam}) in
relation to other MUs and between data sets. Nevertheless, the net
effect of these firing combinations with the estimated expected MUTFs,
see Figure~\ref{fig:RealParam} inserts, does not suggest that the
overall description of the two data sets greatly differ. This
exemplifies the difficulty in disseminating between MUs with similar
excitation and twitch characteristics. The difference in fit could
have occurred in part due to the 70\,mN response level in the R50 data
set not being represented within dataset R10.


\begin{table}
  \begin{center}
  \begin{tabular}{lccccccccc}
    \hline
    Data set												& \multicolumn{4}{c}{R10} 			&$~$& \multicolumn{4}{c}{R50}	\\
    No. of MUs ($u$)								& 7			& 8			& 9			& 10		&		& 7			& 8				& 9			& 10		\\
    \hline
    $\mathbb{P}(u|\mathbf{y})$ (\%)	& 0.04	& 99.95	& 0.01	& 0.00	&		& 0.00	& 100.00	& 0.00	& 0.00	\\
    Grid Size ($n{\times}n$)				& 30		& 30		& 30		& 30		&		& 100		& 90			& 50		& 90		\\
    No. of Particles (,000s)				& 20		& 5			& 5			&	5			&		& 155		& 100			& 65		& 115		\\
    \hline
  \end{tabular}
    \end{center}
  \caption{Posterior summary from the SMC-MUNE procedure for the rat tibial muscle using 10{\micro}sec and 50{\micro}sec duration stimuli.}
  \label{tab:RealResults}
\end{table}

\begin{figure}%
 \includegraphics[width=0.33\textwidth]{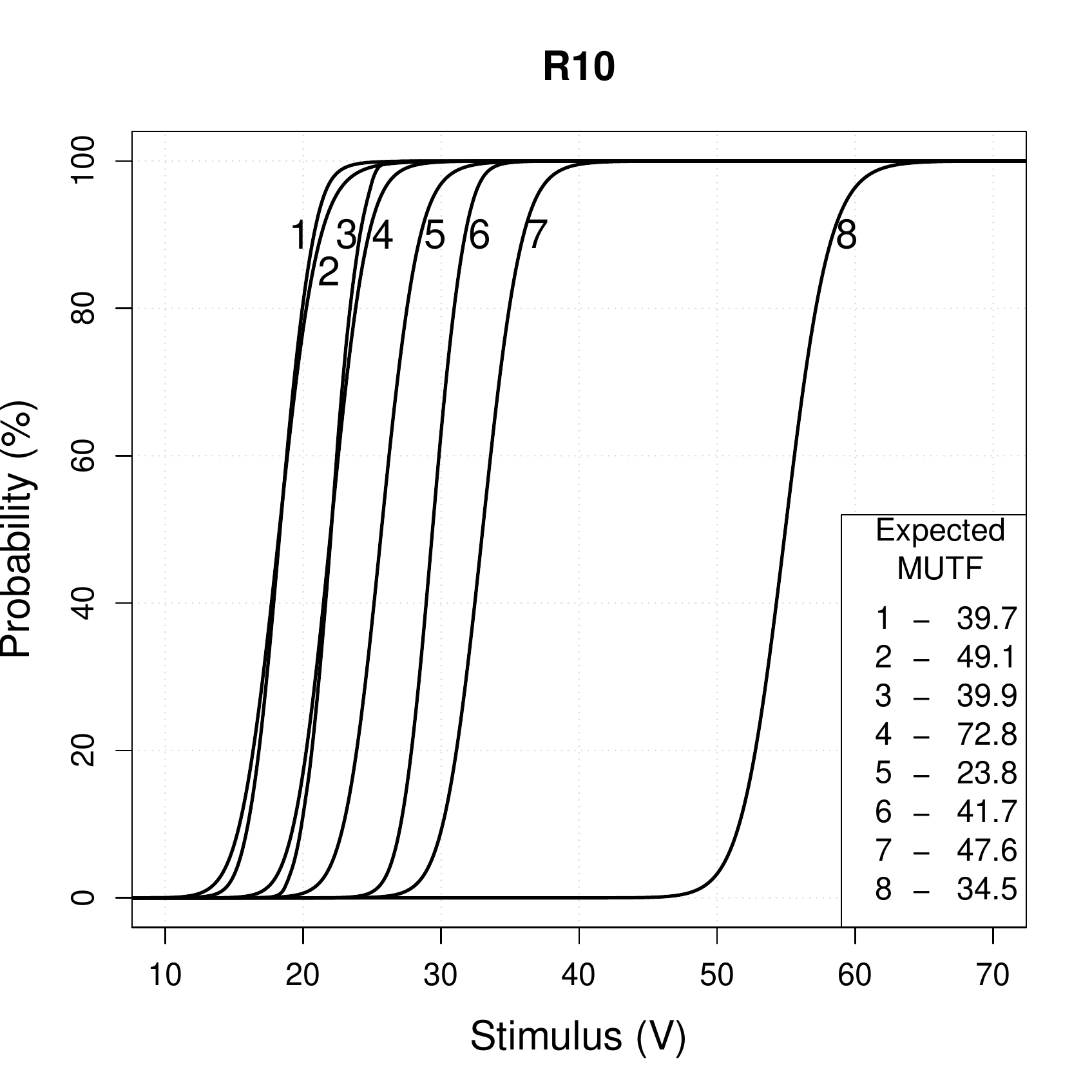}%
 \includegraphics[width=0.33\textwidth]{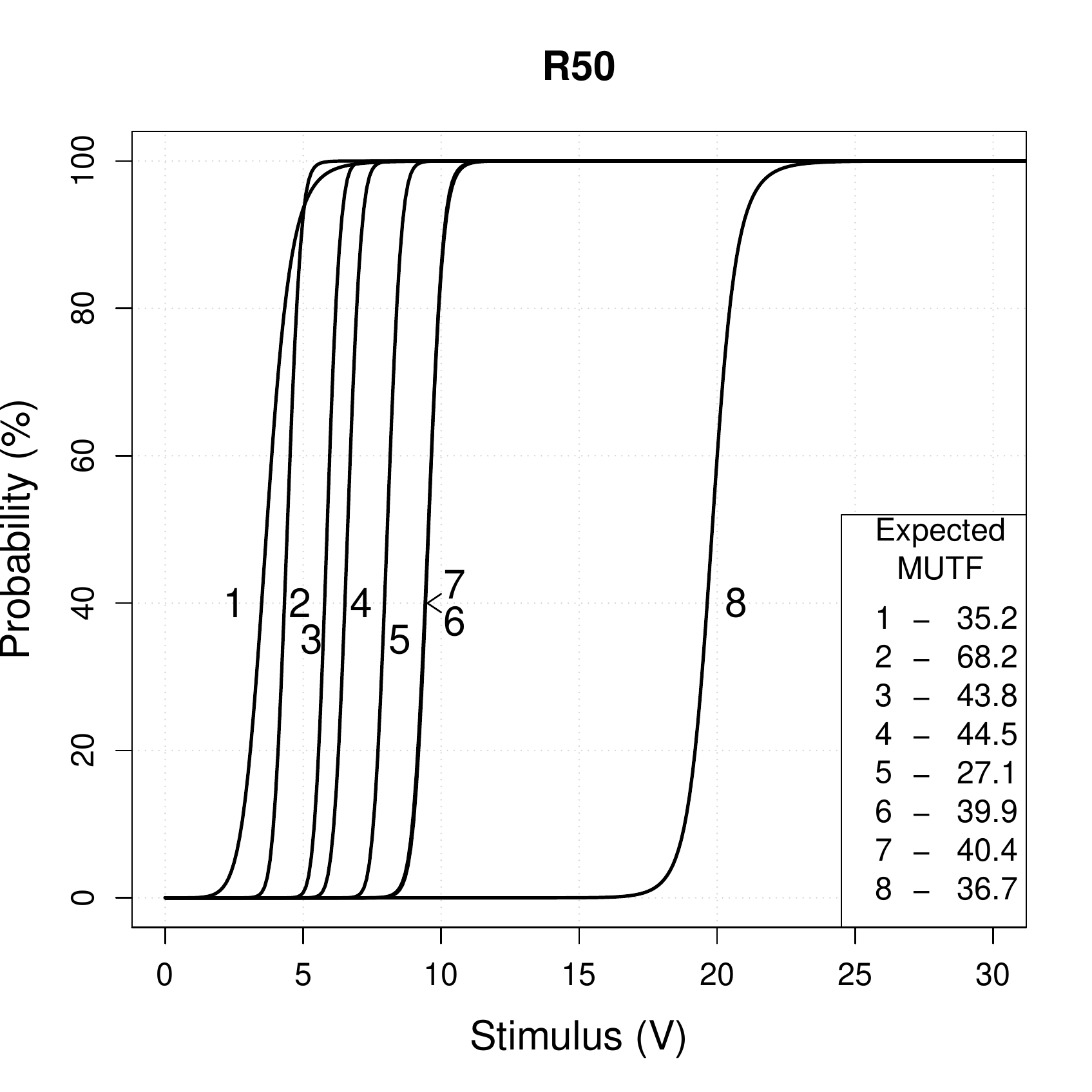}%
 \includegraphics[width=0.33\textwidth]{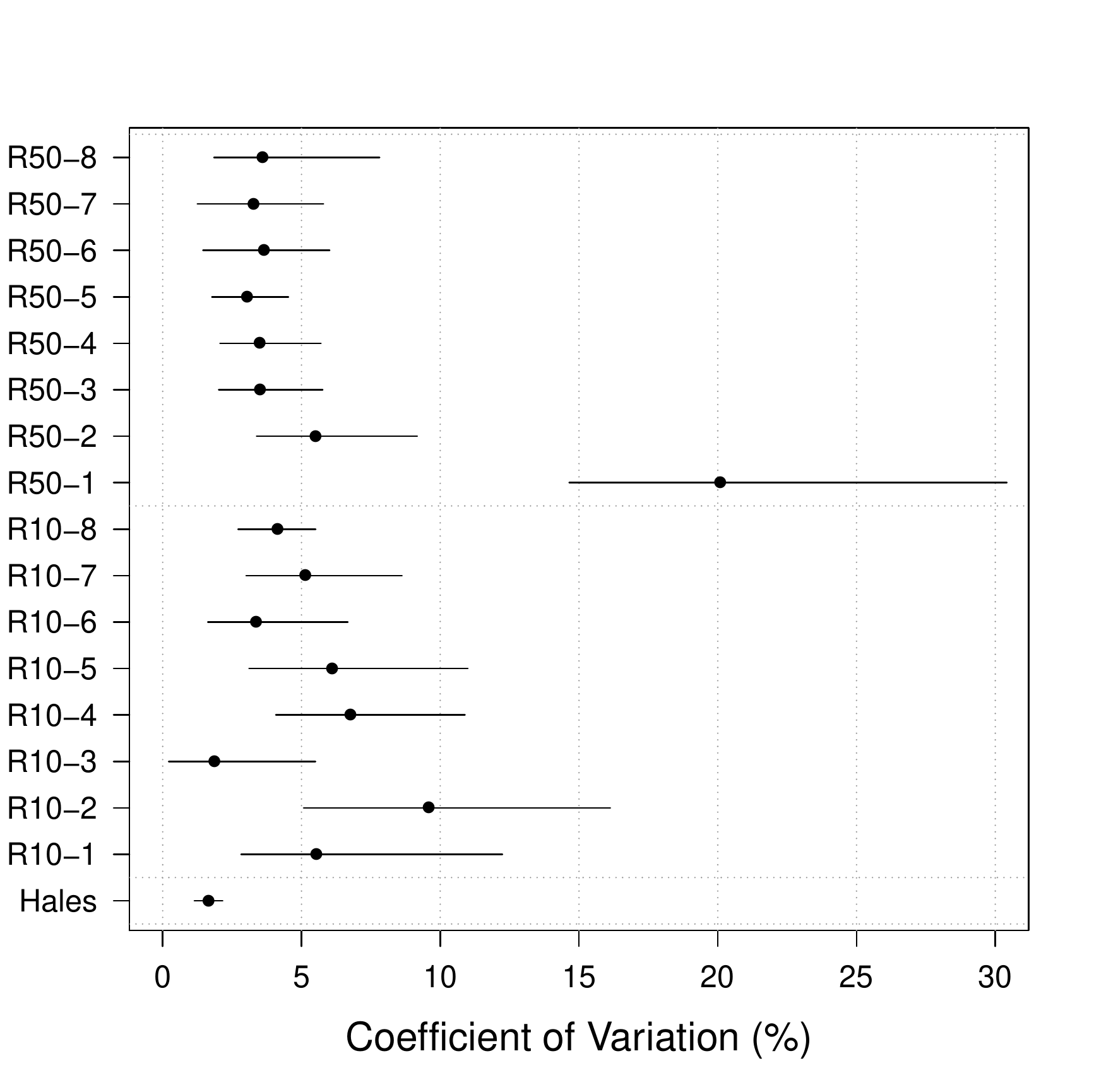}%
\caption{Estimated excitability curves from the eight MU hypotheses
  for data sets R10 (left) and R50 (centre) with corresponding
  expected MUTF mean estimates. Right:
  median and 95\% credible interval for the coefficient of variation for
  the random variable associated with the excitability curve for each
  MU, together with the mean and
  95\% confidence interval from \citeauthor{Hal04}.}%
\label{fig:RealParam}%
\end{figure}

\begin{table}%
\label{tab:Xest}
\caption{Most probable firing events (1=fire, 0=latent) for each level
  in the stimulus-response curve. The labeled MUs for R50 are
  re-ordered to demonstrate similarity between the two data sets. The
  response level around 70\,mN is not present in the R10 data set.}
\begin{center}
\begin{tabular}{lccccccccrlcccccccc}
\hline
					& \multicolumn{8}{c}{R10} 			&$~$& 					 &\multicolumn{8}{c}{R50}	       \\
Level (mN)& 1 & 2 & 3 & 4 & 5 & 6 & 7 & 8 &		& Level (mN) & 1 & 2 & 4 & 3 & 5 & 7 & 6 & 8 \\
\hline
0					& 0 & 0 & 0 & 0 & 0 & 0 & 0 & 0 &		& 0    			 & 0 & 0 & 0 & 0 & 0 & 0 & 0 & 0 \\
50				& 1 & 0 & 0 & 0 & 0 & 0 & 0 & 0 &		& 40 	 		   & 1 & 0 & 0 & 0 & 0 & 0 & 0 & 0 \\
--				&-- &-- &-- &-- &-- &-- &-- &-- & 	& 70 	 		   & 0 & 1 & 0 & 0 & 0 & 0 & 0 & 0 \\
100				& 1 & 1 & 0 & 0 & 0 & 0 & 0 & 0 &		& 110	 		   & 1 & 1 & 0 & 0 & 0 & 0 & 0 & 0 \\
120				& 1 & 0 & 0 & 1 & 0 & 0 & 0 & 0 &		& 120	 		   & 0 & 1 & 0 & 1 & 0 & 0 & 0 & 0 \\
170				& 1 & 1 & 0 & 1 & 0 & 0 & 0 & 0 &		& 150	 		   & 1 & 1 & 0 & 1 & 0 & 0 & 0 & 0 \\
210				& 1 & 1 & 1 & 1 & 0 & 0 & 0 & 0 &		& 200	 		   & 1 & 1 & 1 & 1 & 0 & 0 & 0 & 0 \\
230				& 1 & 1 & 1 & 1 & 1 & 0 & 0 & 0 &		& 230	 		   & 1 & 1 & 1 & 1 & 1 & 0 & 0 & 0 \\
270				& 1 & 1 & 1 & 1 & 1 & 1 & 0 & 0 &		& 270	 		   & 1 & 1 & 1 & 1 & 1 & 1 & 0 & 0 \\
320				& 1 & 1 & 1 & 1 & 1 & 1 & 1 & 0 &		& 300	 		 	 & 1 & 1 & 1 & 1 & 1 & 1 & 1 & 0 \\
360	 			& 1 & 1 & 1 & 1 & 1 & 1 & 1 & 1 &		& 350			   & 1 & 1 & 1 & 1 & 1 & 1 & 1 & 1 \\
\hline
\end{tabular}
\end{center}
\end{table}


\section{Discussion}\label{sec:Discussion}


This paper presents a new sequential Bayesian procedure for motor unit number
estimation (MUNE), the assessment of the number of the operating motor
units (MUs) from an electromyography investigation into
muscle function. 
The fully adpated sequential Monte Carlo (SMC) filter uses the approximate conditional conjugacy of the twitch process.
The
principal purpose of SMC-MUNE is to estimate the marginal likelihood
for the neuromuscular model based on a fixed number of MUs. From this,
motor unit number estimation (MUNE) is then performed by comparing the
evidence between competing MU-number hypotheses. As is demonstrated in Sections
\ref{sec:SimStudyOver} and \ref{sec:SimStudy_under} SMC-MUNE also allows detailed scrutiny of the quality of each model fit.

SMC-MUNE performed well on simulated data, but two scenarios that may
cause incorect estimation were identified. In the first scenario, one
or more MUs were estimated to have a negligible or negative twitch
force, allowing a model that was larger than the truth to fit the data
and resulting in overestimation of the number of motor units
(MUs). This lead to the development of a post-process correction that
restricts the parameter space. By contrast, the second scenario
resulted in under-estimation because of difficulty in estimating the
underlying process during a period of alternation involving many MUs,
where the same stimulus, applied repeatedly can lead to several
different combinations of MUs firing. This issue persisted despite
constraining a key parameter, but was resolved when, instead, additional data points
were sampled from the region of alternation, strongly suggesting that the
original underestimation arose because the information available in
the data was not sufficient to fully characterise the firing process.

Independent application of SMC-MUNE to two data sets
(with data collected as in \citet{Cas10}) on the same
neuromuscular system resulted in the same estimate for the number of MUs. However, closer examination of the model fits identifies minor variations in parameter estimates and firing patterns that reflected subtle and known differences between the two data sets.

The examples investigated in this paper involve neuromuscular systems
with relatively small numbers of MUs. In practice, large and healthy
muscle groups can contain hundreds of MUs \citep{Goo14}. Application of
SMC-MUNE to these larger problems is currently impractical as the
computation cost increases exponentially with the assumed number of
MUs. As such the SMC-MUNE currently is best applied to small neuromuscular systems such as in some animal or in patients with amyotrophic lateral
sclerosis who have limited motor function. 
The computational demand arises from the necessity to evaluate
the predictive mass function for sampling the firing vectors and to
marginalise this event space for calculating the resamping
weights. One approach to address this is to
approximate very low or very high excitation
probabilities by their respective certainties as in \citet{Dro14}. Alternatively, the
excitability curve for SMC-MUNE is specified in generic terms and so
computational saving are possible by defining a function that has
finite support. In addition to concerns over the size of the computation,
additional resources would be required for a sufficiently fine lattice
over the excitability parameter space to minimize numerical error on
the marginal likelihood estimate. Although adaptive sparse grids
\citep{Bun03} have the potential to be more beneficial in terms of
resource management and precision, care would be needed in automating
the grid refinements, and it is likely that a unique grid would be
associated with each distinct firing pattern.

The sequential aspect of the proposed methodology provides the opportunity for real-time inference that has the potential to provide in-lab assistance during experimentation. In this framework, an interim SMC-MUNE analysis could help in identifying the choice of stimulus to apply in order to collect the best evidence to distinguish between competing hypotheses, as in Section~\ref{sec:SimStudy_under}. The limitations of the present SMC-MUNE procedure to become a wholly online algorithm are the computational aspects discussed earlier and the post-processing stage to correct for potentially negligible estimates of the expected MU twitch forces. Solutions to these outstanding problems would increase the efficiency and accuracy of SMC-MUNE and, hence, the range of application.


\section*{Acknowledgements}

The authors would like to thank Dr. Christine Thomas from the Miami Project to Cure Paralysis who provided the data analysed in this paper and for specialist discussions.


\bibliographystyle{plainnat}
\bibliography{MUNErefs}

\begin{thebibliography}{25}
\providecommand{\natexlab}[1]{#1}
\providecommand{\url}[1]{\texttt{#1}}
\expandafter\ifx\csname urlstyle\endcsname\relax
  \providecommand{\doi}[1]{doi: #1}\else
  \providecommand{\doi}{doi: \begingroup \urlstyle{rm}\Url}\fi

\bibitem[Andrieu(2007)]{And07}
C.~Andrieu.
\newblock Discussion of: Motor unit number estimation using reversible jump
  {M}arkov chain {M}onte {C}arlo methods.
\newblock \emph{Journal of the Royal Statistical Society. Series C. Applied
  Statistics}, 56\penalty0 (3):\penalty0 261--263, 2007.

\bibitem[Bostock(1983)]{Bos83}
H~Bostock.
\newblock The strength-duration relationship for excitation of myelinated
  nerve: computed dependence on membrane parameters.
\newblock \emph{The Journal of Physiology}, 341\penalty0 (1):\penalty0 59--74,
  1983.

\bibitem[Bromberg(2003)]{Bro03}
M.~B. Bromberg.
\newblock Consensus.
\newblock \emph{Supplements to Clinical Neurophysiology}, 55\penalty0
  (C):\penalty0 333--338, 2003.

\bibitem[Bromberg(2007)]{Bro07}
M.~B. Bromberg.
\newblock Updating motor unit number estimation ({MUNE}).
\newblock \emph{Clinical neurophysiology}, 118\penalty0 (1):\penalty0 1--8,
  2007.

\bibitem[Brown and Milner-Brown(1976)]{Bro76}
W.~F. Brown and H.~S. Milner-Brown.
\newblock Some electrical properties of motor units and their effects on the
  methods of estimating motor unit numbers.
\newblock \emph{Journal of Neurology, Neurosurgery \& Psychiatry}, 39\penalty0
  (3):\penalty0 249--257, 1976.

\bibitem[Bungartz and Dirnstorfer(2003)]{Bun03}
H.-J. Bungartz and S.~Dirnstorfer.
\newblock Multivariate quadrature on adaptive sparse grids.
\newblock \emph{Computing}, 71\penalty0 (1):\penalty0 89--114, 2003.

\bibitem[Carvalho et~al.(2010)Carvalho, Johannes, Lopes, and Polson]{Car10}
C.~M. Carvalho, M.~S. Johannes, H.~F. Lopes, and N.~G. Polson.
\newblock Particle learning and smoothing.
\newblock \emph{Statistical Science. A Review Journal of the Institute of
  Mathematical Statistics}, 25\penalty0 (1):\penalty0 88--106, 2010.

\bibitem[Casella et~al.(2010)Casella, Almeida, Grumbles, Liu, and
  Thomas]{Cas10}
G.~T.~B. Casella, V.~W. Almeida, R.~M. Grumbles, Y.~Liu, and C.~K. Thomas.
\newblock Neurotrophic factors improve muscle reinnervation from embryonic
  neurons.
\newblock \emph{Muscle \& nerve}, 42\penalty0 (5):\penalty0 788--797, 2010.

\bibitem[Daube(1995)]{Dau95}
J.~R. Daube.
\newblock Estimating the number of motor units in a muscle.
\newblock \emph{Journal of Clinical Neurophysiology}, 12\penalty0 (6):\penalty0
  585--594, 1995.

\bibitem[Drovandi et~al.(2014)Drovandi, Pettitt, Henderson, and McCombe]{Dro14}
C.~C. Drovandi, A.~N. Pettitt, R.~D. Henderson, and P.~A. McCombe.
\newblock Marginal reversible jump {M}arkov chain {M}onte {C}arlo with
  application to motor unit number estimation.
\newblock \emph{Computational Statistics \& Data Analysis}, 72:\penalty0
  128--146, 2014.

\bibitem[Genz and Bretz(2009)]{Genz09}
A.~Genz and F.~Bretz.
\newblock \emph{Computation of multivariate normal and {$t$} probabilities},
  volume 195 of \emph{Lecture Notes in Statistics}.
\newblock Springer, Dordrecht, 2009.

\bibitem[Gooch et~al.(2014)Gooch, Doherty, Chan, Bromberg, Lewis, Stashuk,
  Berger, Andary, and Daube]{Goo14}
C.~L. Gooch, T.~J. Doherty, K.~M. Chan, M.~B. Bromberg, R.~A. Lewis, D.~W.
  Stashuk, M.~J. Berger, M.~T. Andary, and J.~R. Daube.
\newblock Motor unit number estimation: a technology and literature review.
\newblock \emph{Muscle \& nerve}, 50\penalty0 (6):\penalty0 884--893, 2014.

\bibitem[Green(1995)]{Gre95}
P.~J. Green.
\newblock Reversible jump {M}arkov chain {M}onte {C}arlo computation and
  {B}ayesian model determination.
\newblock \emph{Biometrika}, 82\penalty0 (4):\penalty0 711--732, 1995.

\bibitem[Hales et~al.(2004)Hales, Lin, and Bostock]{Hal04}
J.~P. Hales, C.~S.-Y. Lin, and H.~Bostock.
\newblock Variations in excitability of single human motor axons, related to
  stochastic properties of nodal sodium channels.
\newblock \emph{The Journal of physiology}, 559\penalty0 (3):\penalty0
  953--964, 2004.

\bibitem[Henderson et~al.(2006)Henderson, Ridall, Pettitt, McCombe, and
  Daube]{Hen06}
R.~D. Henderson, G.~R. Ridall, A.~N. Pettitt, P.~A. McCombe, and J.~R. Daube.
\newblock The stimulus--response curve and motor unit variability in normal
  subjects and subjects with amyotrophic lateral sclerosis.
\newblock \emph{Muscle \& nerve}, 34\penalty0 (1):\penalty0 34--43, 2006.

\bibitem[Hol et~al.(2006)Hol, Schon, and Gustafsson]{Hol06}
J.~D. Hol, T.~B. Schon, and F.~Gustafsson.
\newblock On resampling algorithms for particle filters.
\newblock In \emph{Nonlinear Statistical Signal Processing Workshop, 2006
  IEEE}, pages 79--82. IEEE, 2006.

\bibitem[Kadrie et~al.(1976)Kadrie, Yates, Milner-Brown, and Brown]{Kad76}
H.~A. Kadrie, S.~K. Yates, H.~S. Milner-Brown, and W.~F. Brown.
\newblock Multiple point electrical stimulation of ulnar and median nerves.
\newblock \emph{Journal of Neurology, Neurosurgery \& Psychiatry}, 39\penalty0
  (10):\penalty0 973--985, 1976.

\bibitem[Major and Jones(2005)]{Maj05}
L.~A. Major and K.~E. Jones.
\newblock Simulations of motor unit number estimation techniques.
\newblock \emph{Journal of neural Engineering}, 2\penalty0 (2):\penalty0 17,
  2005.

\bibitem[Major et~al.(2007)Major, Hegedus, Weber, Gordon, and Jones]{Maj07}
L.~A. Major, J.~Hegedus, D.~J. Weber, T.~Gordon, and K.~E. Jones.
\newblock Method for counting motor units in mice and validation using a
  mathematical model.
\newblock \emph{Journal of neurophysiology}, 97\penalty0 (2):\penalty0
  1846--1856, 2007.

\bibitem[McComas et~al.(1971)McComas, Fawcett, Campbell, and Sica]{McC71}
A.~J. McComas, P.~Fawcett, M.~J. Campbell, and R.~E.~P. Sica.
\newblock Electrophysiological estimation of the number of motor units within a
  human muscle.
\newblock \emph{Journal of Neurology, Neurosurgery \& Psychiatry}, 34\penalty0
  (2):\penalty0 121--131, 1971.

\bibitem[Pitt and Shephard(1999)]{Pit99}
M.~K. Pitt and N.~Shephard.
\newblock Filtering via simulation: auxiliary particle filters.
\newblock \emph{Journal of the American Statistical Association}, 94\penalty0
  (446):\penalty0 590--599, 1999.

\bibitem[Ridall et~al.(2006)Ridall, Pettitt, Henderson, and McCombe]{Rid06}
P.~G. Ridall, A.~N. Pettitt, R.~D. Henderson, and P.~A. McCombe.
\newblock Motor unit number estimation---a {B}ayesian approach.
\newblock \emph{Biometrics}, 62\penalty0 (4):\penalty0 1235--1250, 2006.

\bibitem[Ridall et~al.(2007)Ridall, Pettitt, Friel, McCombe, and
  Henderson]{Rid07}
P.~G. Ridall, A.~N. Pettitt, N.~Friel, P.~A. McCombe, and R.~D. Henderson.
\newblock Motor unit number estimation using reversible jump {M}arkov chain
  {M}onte {C}arlo methods.
\newblock \emph{Journal of the Royal Statistical Society. Series C. Applied
  Statistics}, 56\penalty0 (3):\penalty0 235--269, 2007.

\bibitem[Shefner et~al.(2006)Shefner, Cudkowicz, and Brown]{She06}
J.~M. Shefner, M.~Cudkowicz, and R.~H. Brown.
\newblock Motor unit number estimation predicts disease onset and survival in a
  transgenic mouse model of amyotrophic lateral sclerosis.
\newblock \emph{Muscle \& nerve}, 34\penalty0 (5):\penalty0 603--607, 2006.

\bibitem[Stashuk et~al.(1994)Stashuk, Doherty, Kassam, and Brown]{Sta94}
D.~W. Stashuk, T.~J. Doherty, A.~Kassam, and W.~F. Brown.
\newblock Motor unit number estimates based on the automated analysis of
  {F}-responses.
\newblock \emph{Muscle \& nerve}, 17\penalty0 (8):\penalty0 881--890, 1994.

\end{thebibliography}

\appendix 

\section{Additional detail}
\label{sec:APX}
The prior sufficient statistics for the simulation and case studies are: $\mb_0 = 0$, $\cb_0 = 10^3$, $\ab_0 = 0.5$, $\bb_0 = 0.1$, $\M_0 = 40 \unit_u$, $C_0 = 10^4I_u$ and $a_0 = 0.5$ where $\unit_u$ is a unit $u$-vector and $I_u$ is the $u\times{u}$ identity matrix. The statistic $b_0$ is defined according to \eqref{eq:DefineNuPrior} with $\delta=0.05$ and $\epsilon=0.2$. The upper bounds for the excitability parameter space are $\etamax = 1.1s_{\tau}$ and $\lambdamax = 14$. In the case study, the upper bound for the scale parameter was reduced to $\lambdamax=7$.

Resampling in Algorithm~\ref{tab:Alg}, is performed by systematic sampling on the residuals of particle weights \citep{Hol06}. The number of particle samples and the number of rectangular lattice cells  is initially $N=5000$ and $|\mathcal{G}|=30\times{30}$ respectively. 

Accuracy in MU-number posterior mass function is managed by ensuring
that for each model, $u$, the range in marginal log-likelihood
estimate from 3 independent runs of the SMC scheme is less than $1$
whenever the posterior probability is greater than 1\%. If not, then
the particle set is increased in steps of $5000$ samples to reduce
Monte Carlo variability. Once this criterion is satisfied, the lattice
for the numerical integration 
is made finer by 10 vertices in both dimensions and the stability of
the estimates to increasing grid size is checked; instability leads to
a further check of the Monte Carlo variability and, if necessary,
an increase in the particle set size, and then a further increase in
the number of vertices; iteration between these two steps continues
until the results are numerically stable and have a low variance. For
a particular data set, once
the minimal number of particles and grid size required for stability
have been ascertained, a further $7$ runs are performed using these
settings and the final marginal log-likelihood estimate is the average of
the result from the total of $10$ runs.

\end{document}